\documentclass[a4paper,11pt]{article}

\usepackage{graphicx}
\usepackage{authblk}
\usepackage{fullpage}
\usepackage{libertine}
\usepackage{color}

\usepackage[ruled,vlined]{algorithm2e}
\SetArgSty{textrm}

\usepackage{amsmath,amsfonts,amssymb,amsthm}

\usepackage[breaklinks=true]{hyperref}
\usepackage[svgnames]{xcolor}
\usepackage[capitalise,nameinlink]{cleveref}
\hypersetup{colorlinks={true},linkcolor={DarkBlue},citecolor=[named]{DarkGreen}}

\usepackage{natbib}

\usepackage{subcaption}
\usepackage{caption}

\usepackage{booktabs} 
\usepackage[ruled]{algorithm2e} 
\usepackage{enumerate, todonotes}

\SetAlFnt{\small}
\SetAlCapFnt{\small}
\SetAlCapNameFnt{\small}
\SetAlCapHSkip{0pt}
\IncMargin{-\parindent}

\usepackage{bbm}

\usepackage{tikz}
\usetikzlibrary{fit,shapes.geometric}

\newcounter{nodemarkers}

\usetikzlibrary{fit,shapes.misc}

\newtheorem{theorem}{Theorem}[section]
\newtheorem{defn}[theorem]{Definition}

\newtheorem{lemma}[theorem]{Lemma}

\newcommand{\bigO}{\mathcal{O}}
\newcommand{\agents}{\mathcal{N}}
\newcommand{\items}{\mathcal{I}}
\newcommand{\val}{\mathcal{V}}
\newcommand{\Q}{\mathbf{Q}}
\newcommand{\x}{\mathbf{x}}
\newcommand{\p}{\mathbf{p}}
\newcommand{\tbold}{\mathbf{t}}
\newcommand{\E}{\mathbb{E}}
\newcommand{\M}{\mathcal{M}}
\newcommand{\K}{\mathcal{K}}
\newcommand{\X}{\mathcal{X}}
\newcommand{\sw}{{\sf SW}}
\newcommand{\opt}{\text{OPT}}
\newcommand{\util}{{\sf U}}
\newcommand{\egal}{{\sf E}}

\newcommand{\lognash}{{\sf lgN}}
\newcommand{\avnash}{{\sf avN}}

\DeclareMathOperator{\prm}{Pr}
\DeclareMathOperator{\EG}{iEG}
\DeclareMathOperator*{\argmax}{arg\,max}
\newtheorem{example}[theorem]{Example}

\allowdisplaybreaks
\begin{document}
\title{On Interim Envy-Free Allocation Lotteries}

\author[1]{Ioannis Caragiannis}


\author[2]{Panagiotis Kanellopoulos}
\author[2]{Maria Kyropoulou}
\affil[1]{Department of Computer Science, Aarhus University, Denmark}
\affil[2]{School of Computer Science and Electronic Engineering, University of Essex, UK}

\date{}
\maketitle

\begin{abstract}
With very few exceptions, recent research in fair division has mostly focused on deterministic allocations. Deviating from this trend, we study the fairness notion of {\em interim envy-freeness} (iEF) for lotteries over allocations, which serves as a sweet spot between the too stringent notion of ex-post envy-freeness and the very weak notion of ex-ante envy-freeness. Our analysis relates iEF to other fairness notions as well, and reveals tradeoffs between iEF and efficiency. Even though several of our results apply to general fair division problems, we are particularly interested in instances with equal numbers of agents and items where allocations are perfect matchings of the items to the agents. We show how to compute iEF allocations in matching allocation instances in polynomial time, while also maximizing several efficiency objectives, even though this proves to be considerably more challenging than computing envy-free allocations. Our algorithms use efficient solutions to a novel variant of the bipartite matching problem. We also study the extension of iEF when payments to or from the agents are allowed. We present a series of results on two optimization problems, including a generalization of the classical rent division problem to random~allocations. 
\end{abstract}


%


%
%
%

\section{Introduction}
Plenty of situations arise in the real world every day in which assets need to be distributed among individuals. Making sure that everyone gets what they are entitled to is an imperative, yet vague, aspiration that is open to interpretation. Fair division is a research area that deals with problems of distributing assets in a way that is considered fair. Fair allocation problems, which focus on indivisible items, have received considerable attention from the EconCS community recently.
 
Among the fairness notions that have been proposed with the goal of capturing the necessity for impartiality and justice, envy-freeness~\citep{F67,V74} is, without doubt, the prevailing one. Envy-freeness requires that each individual, or agent, prefers their own share to anyone else's. However natural and intuitive, though, envy-freeness may not be possible to achieve. In addition, the universality of fair division disputes justifies many different definitions of fairness. Some additional popular fairness notions in the literature include proportionality and max-min fair share, among others. 
 
The vast majority of the related literature on envy-freeness focuses on deterministic allocations, while for random allocations (lotteries or probability distributions over allocations) the relevant fairness concepts are those of ex-ante and ex-post envy-freeness. 
%
Ex-ante envy-freeness compares the random bundle allocated to an agent, in terms of expected valuation, to the random bundle allocated to any other agent and is very weak as a fairness guarantee. Indeed, a lottery that allocates all items to an agent selected uniformly at random is ex-ante envy-free; clearly, such a lottery can hardly be considered fair. On the other extreme, the notion of ex-post envy-freeness requires that every outcome of a random allocation is envy-free. Ex-post envy-freeness is very strict and essentially invalidates the advantages of randomness. 

In an attempt to overcome the deficiencies of the ex-ante and ex-post fairness guarantees mentioned above, a recent stream of research on lotteries considers them simultaneously instead of separately. Such guarantees are referred to as “best-of-both-worlds” guarantees \cite{AFSV23}, and they examine whether it is possible to achieve desirable ex-ante outcomes by randomizing over acceptable deterministic outcomes, thus also ensuring ex-post guarantees. For example, \cite{AFSV23} showed that there always exist ex-ante envy-free lotteries over envy-free up to one item (EF1, a relaxation of envy-freeness) deterministic allocations. Such works pursue the approaches of randomization and relaxation simultaneously by considering lotteries that are considered fair both before and after their implementation.

The notion of interim envy-freeness \citep*{PS87}, presented in this paper, provides an alternative but related approach to strengthening ex-ante fairness guarantees. Conceptually, without resorting to relaxations of ex-post fairness, like EF1, interim envy-freeness assumes partial (but consistent) uncertainty about the realization of the ex-ante lottery while maintaining the desire for envy-freeness conceptually. It serves as a middle ground between ex-post and ex-ante envy-freeness, balancing the stringency of the
constraint and the substance of the fairness guarantee. Interestingly, interim envy-freeness also achieves proportionality ex-post. 

\begin{example}\label{ex:iEF-noEF}
\begin{table}[h!]
\centering
\caption{The instance used in Example \ref{ex:iEF-noEF}. \label{tab:ef-vs-ief}}
{\begin{tabular}{c| c cc}
\hline
\phantom{aaa} &\phantom{aaa}$a$&\phantom{aaa}$b$&\phantom{aaa}$c$\\
\hline
\phantom{aaa}$1$&\phantom{aaa}$1/3$&\phantom{aaa}$2/3$&\phantom{aaa}$0$\\
\phantom{aaa}$2$&\phantom{aaa}$0$&\phantom{aaa}$1/2$&\phantom{aaa}$1/2$\\
\phantom{aaa}$3$&\phantom{aaa}$0$&\phantom{aaa}$1/2$&\phantom{aaa}$1/2$\\
\hline
\end{tabular}}
{}
\end{table}

Consider the asset allocation instance of Table~\ref{tab:ef-vs-ief} 
where the set $\{a, b, c\}$ of indivisible items are to be allocated to the set $\{1, 2, 3\}$ of agents.  
 A concise notation like $a$-$c$-$b$ is used to represent the allocation in which agents $1$, $2,$ and $3$ get items $a$, $c$, and $b$, respectively.

Now, consider the fixed, publicly known lottery that returns allocations $a$-$b$-$c$ and $a$-$c$-$b$, with probability $1/2$ each. 
Let an outcome of the lottery be realized and each agent observe only their own allocation. Each agent, then, compares their utility to the random bundle allocated to any other agent, conditioned upon their own realized allocation. If such a constraint is satisfied for every agent, with respect to every other agent and any possible bundle, then the lottery is said to be interim envy-free. 

Note that agent $1$ is allocated item $a$, for which she has a value of $1/3$, in both allocations in the support of the lottery. Agent $2$ (and, similarly, agent $3$) gets item $b$ with probability $1/2$ and item $c$ with probability $1/2$. Agent $1$'s average value for the item agent $2$ (or agent $3$) gets is $2/3\cdot 1/2 + 0\cdot 1/2 = 1/3$. Hence, agent $1$ does not envy the average bundle of any other agents conditional on the fact that she received item $a$; we say that the interim envy-freeness condition for agent $1$ with respect to agent $2$ and bundle $\{a\}$ (and, similarly, for agent $1$ with respect to $3$ and bundle $\{a\}$) is satisfied. Since $1$ never receives a different bundle than $\{a\}$, the interim envy-freeness condition for agent $1$ with respect to both other agents is satisfied. Arguing about the interim envy-free conditions of agents $2$ and $3$ is easier as, in our example, agents $2$ and $3$ have maximum possible value in each of the two allocations produced by the lottery and, hence, they do not envy any other agent; the (interim) envy-freeness conditions are satisfied for them as well. We can conclude that this is an interim envy-free lottery. 
\end{example}

One might observe an analogy between the notions of interim envy-freeness in a resource allocation context, and that of Bayesian equilibrium in a game-theoretic model where a set of players strategize to maximize their individual utilities. In the latter, there is a common prior on the strategy profiles of all players, and each player selects an optimal strategy given her beliefs about the other players’ strategies, similar to a game of poker. In other words, the notion of a Bayesian equilibrium involves probabilistic reasoning about the states of nature, which are captured by the potential strategy profiles of players. In a resource allocation setting, the notion of interim envy-freeness similarly involves probabilistic reasoning about the states of nature, only now, the states of nature are captured by the potential allocations of resources to agents. In particular, the players estimate the split of resources in the form of a probability distribution over allocations. This probability distribution is common knowledge, but individual players update their estimation through private information of their personal allocation once that is realized. As a consequence, the differences in the players’ perception are entirely attributed to the differences in private information on top of the common prior assumption they all have, similar to the notion of Bayesian equilibria. Such notions are more general (and weaker) than their deterministic counterparts, in the sense that any property that is satisfied in every possible state of nature will be straightforwardly satisfied in these probabilistic extensions. 

The interim aspect limits this generality by capturing the intuitive assumption that in an uncertain environment, it is easier for an agent to understand their own value (or allocation) than that of others. This notion helps avoid trivial shortcomings of worst-case-type analysis in a meaningful way. To see that, note that there are resource allocation instances that do not admit any (ex-post) envy-free allocation but do admit an interim envy-free allocation. Example \ref{ex:iEF-noEF} demonstrates such an instance; 
to see that there is no EF allocation, note that each agent must obtain exactly one item and agent $1$ must therefore obtain item $b$, which leaves the agent among $2$ and $3$ who obtains item $a$ envious of both other agents. We remark, furthermore, that the ex-post guarantee of EF1 is trivially satisfied, and is hence meaningless, when each agent may obtain at most one item.

Interim envy-freeness can be naturally extended to accommodate payments, similarly to the recent fair allocation literature~\citep*{GMPZ17,HS19}. It is known that payments can help eliminate envy in the deterministic allocation case, both in the form of subsidies paid to the agents to compensate for an unsatisfactory bundle, and in the form of rent payments paid by the agents to make their allocation look less desirable to others. Rent division~\citep*{ASU04} is a fundamental fair allocation problem involving payments, where the input consists of a total rent amount, a set of agents, and an equal number of rooms on which the preferences of agents are expressed. The goal is to assign a price to each room so that the room prices sum up to the total rent, and to match agents to rooms so that everyone prefers their own allocation and rent share. Using the interim envy-freeness concept, we consider natural extensions of problems with payments to the random allocation case, both in the rent division and in the subsidy distribution context. 
 
Matching allocation instances, as in the rent division setting just discussed, are relevant in many applications; hence, we partially focus on this case. An important technical advantage is that such instances allow for an easy computation of (deterministic) envy-free allocations, as opposed to general allocation instances, for which relevant problems are typically NP-hard.
However, allowing randomization seems to make the situation considerably more complex. Interestingly, as we will see, the added complication still allows for positive computational results related to interim envy-free lotteries. 
 
\subsection{Overview and significance of our contribution}
To the best of our knowledge, interim envy-freeness (iEF) has not received any attention by the EconCS community. 
Our work presents several key observations into this previously unexplored notion, highlighting its advantages over alternative fairness concepts. We begin by summarizing important properties that emphasize the significance of interim envy-freeness as a fairness notion for lotteries of allocations, followed by a more thorough discussion of our results. Most importantly, and aligned with the ``best-of-both-wolds'' paradigm discussed above, interim envy-freeness is a probabilistic notion that provides substantial fairness guarantees for each deterministic allocation that can be realized. In particular, an interim envy-free lottery only has proportional\footnote{An allocation to $n$ agents is proportional if each allocated bundle is worth to its corresponding agent at least a $1/n$ fraction of the value of the total set of items.} allocations in its support. This has several implications, including bounding the maximum envy at any allocation in the support of an interim envy-free lottery. This is clearly not the case for general ex-ante envy-free allocations: a trivial such allocation would allocate the entire bundle to each agent with equal probability thus violating proportionality and resulting in unbounded envy in each realization. Furthermore, an interim envy-free lottery always satisfies ex-ante envy-freeness, while interim envy-freeness is easier to satisfy than ex-post envy-freeness (as explained above). When extended to settings with payments, interim envy-free lotteries can be superior to envy-free allocations with respect to objectives relating to payments and utility.

We continue with a more detailed summary of the interesting properties that interim envy-freeness enjoys. First, we relate it to the most important fairness properties for deterministic allocations and lotteries. In terms of strength as a fairness property, iEF is proved to lie between proportionality and envy-freeness in the following way. Clearly, when viewed as a degenerate lottery, any envy-free allocation is iEF. Also, every iEF lottery is defined over proportional allocations (Lemma \ref{lem:ief-implies-prop}). These implications are shown to be strict in a strong sense. We show that there are allocation instances that admit proportional allocations but no iEF lottery (Lemma \ref{lem:prop-no-ief}), and instances that admit iEF lotteries but no envy-free allocation (Lemma \ref{lem:ief-vs-ef}). Compared to fairness properties for lotteries, iEF lies between ex-ante envy-freeness (which can be always attained trivially) and ex-post envy-freeness (which is too restrictive) (Lemmas \ref{lem:ex-ante} and \ref{lem:ex-post}). These findings and observations appear in Section~\ref{sec:ief} alongside other results relating iEF with other fairness properties.

Our next goal is to explore the trade-offs between iEF and economic efficiency (Section~\ref{sec:ief-vs-efficiency}). We pay special attention to matching allocation instances where envy-freeness implies Pareto-efficiency and optimal utilitarian, egalitarian, and average Nash social welfare. A careful interpretation of these facts reveals that envy-freeness is a very restrictive fairness property. In contrast, being less restrictive, iEF lotteries may not be Pareto-efficient (Theorem \ref{thm:po-3}) and can furthermore produce allocations of suboptimal social welfare. We provide tight or almost tight bounds on the price of iEF with respect to the utilitarian, egalitarian, and average Nash social welfare (Theorems \ref{thm:poief-upper}, \ref{thm:poief-util-lower}, \ref{thm:poief-egal-lower}, and \ref{thm:poief-nash-lower}). These bounds suggest that iEF allocations can have social welfare that is up to $\Theta(n)$ times far from optimal, where $n$ is the number of agents.

Bounds on the price of iEF give only rough estimates of the best social welfare of iEF lotteries. We present polynomial-time algorithms for computing iEF lotteries that maximize the utilitarian, egalitarian, and log-Nash social welfare (Theorem \ref{thm:compute-ief}). Our algorithms follow a general template that can be briefly described as follows. The problem of computing an iEF lottery of maximum social welfare is formulated as a linear program. This linear program has exponentially many variables; to solve it, our algorithm instead solves the dual linear program using the ellipsoid method. As the dual linear program has exponentially many constraints, the ellipsoid method needs access to polynomial-time separation oracles that check whether the dual variables violate the dual constraints or not. We design such separation oracles by exploiting connections to maximum edge-pair-weighted bipartite perfect matching (2EBM), a novel (to the best of our knowledge) combinatorial optimization problem that involves perfect matchings in bipartite graphs. We show how to solve 2EBM in polynomial time by exploiting an elegant lemma by~\citet{C77} on decompositions of doubly-stochastic centro-symmetric matrices. We believe that 2EBM is a natural combinatorial optimization problem of independent interest and with applications in other contexts. These computational results appear in Section~\ref{sec:computing} and constitute the most technically intriguing results in the paper.

Finally, we extend the definition of interim envy-freeness to accommodate settings where monetary transfers (payments) are allowed. We define and study two related optimization problems. In subsidy minimization, which is motivated by a similar problem for deterministic allocations that was studied recently (see e.g.,~\citep*{BDN+20, CI20, HS19}), we seek iEF pairs of lotteries and payments to the agents so that the total expected amount of payments is minimized. In utility maximization, which extends the well-known rent division problem, we seek iEF pairs of lotteries and payments that are collected from the agents and contribute to a fixed rent; the objective is to maximize the minimum expected agent utility. We consider different types of payments depending on whether the payments are agent-specific, bundle-specific, or unconstrained (i.e., specific to agents and allocations). iEF is proved to be considerably more powerful than envy-freeness, allowing for much better solutions to the two optimization problems compared to their deterministic counterparts (Example \ref{ex:ief-vs-ef}). We also showcase the importance of both agent-specific and bundle-specific payments by showing that they are incomparable to each other, in the context of the two optimization problems (Theorems \ref{thm:a-vs-b} and \ref{thm:b-vs-a}). By applying our computational template, we present efficient algorithms that compute optimal solutions to subsidy minimization (Theorem \ref{thm:sm}) and utility maximization (Theorem \ref{thm:umax}) using unconstrained payments, violating the iEF condition only marginally. These results are presented in Section~\ref{sec:payments}. We believe that they also attest to the significance of interim envy-freeness and will motivate further study.

\subsection{Further related work}
Previous work on randomness in allocation problems is clearly related to ours. \citet{Aziz19} discusses the benefits of using randomization in social choice settings, including fair division, and addresses the related challenges. \citet{GT02} study ex-ante and ex-post notions of fairness in a setting with inherent uncertainty imposed by the environment. More recently, \citet{AAGW15} analyze the ex-post and ex-ante envy-freeness guarantees of particular algorithms for online fair division. \citet{AFSV23} study the possibility of achieving ex-ante and ex-post fairness guarantees simultaneously in the classical fair allocation setting. For example, they show that there always exists an ex-ante envy-free lottery, with all allocations in its support satisfying a relaxed fairness property known as envy-freeness up to one item (EF1; see~\citep{B11}). They furthermore show that such lotteries can be computed efficiently. \citet{BEF22} present a polynomial time algorithm that computes a lottery that is ex-ante proportional while every ex-post allocation gives every agent at least half of her maximin share (MMS; see again \citep{B11}). Very recently, \citet{AGM23} and \citet{HSV23} extend the study of simultaneously achieving ex-ante and ex-post fairness guarantees to the case of agents with different entitlements.

Interim envy-freeness was first defined by \citet{PS87} (see also \citep*{Ves04,deC08}).
\citet{PS87} consider a differential information economy with a set of states, a set of agents, and a divisible good (endowment), where an agent's utility for a given amount of endowment depends also on the state. Each agent has a prior probability distribution over the states and is aware of a partition of the set of states. Partitions and priors are public information and agent-specific, while the outcome of the partition that each agent observes is private information. Among other questions, they study whether there exist truthful mechanisms leading to allocations of the endowment that satisfy interim envy-freeness. \citet{deC08} considers a similar, albeit non-strategic, setting with multiple items where the aggregate endowment is state-dependent and shows that interim envy-free allocations do not always exist. \citet{Ves04} discusses a setting with an equal number of agents and items, allowing also for a divisible good in the form of money, and studies existence of interim envy-free allocations when each agent must obtain exactly one item. 
 Even though the intuition behind the iEF notion in those papers coincides with ours, their settings are different. They are more complex in the sense that the value of an agent for a bundle may depend on the allocation, but they are more restrictive as the lottery probabilities are fixed in advance. Our setting is more suitable to formulate and study existence and computational questions. 

Randomness is an important design tool for mechanisms that compute solutions in matching allocation instances. In a slightly different context than ours~(e.g., see~\citep*{HZ79}), agents are assumed to have ordinal preferences on the items. Mechanisms such as the probabilistic serial rule~\citep*{BM01} introduce randomness to avoid discrimination between agents, and achieve ex-ante fairness guarantees for all cardinal valuations that are compatible to the ordinal preferences. Other rules in this line of research include random priority~\citep*{AS98}, vigilant eating~\citep*{AB20}, and several extensions of the probabilistic serial rule~\citep*{BCKM13}.

Previous work on subsidy minimization has focused on envy-freeness as the solution concept. The objective is to compute an allocation and appropriate payments to the agents so that agents are non-envious for the combinations of payments and bundles of items they receive. The work of \citet{Maskin1987} seems to be the first treatment of the problem, although without optimizing subsidies. The notion of envy-freeability refers to an allocation that can become envy-free when paired with appropriate payments. \citet{HS19} present a characterization that indicates that envy-freeability is strongly connected to a no-positive-cycle property in an appropriately defined envy graph. Alternatively, envy-freeable allocations maximize the utilitarian social welfare with respect to all bundle redistributions among the agents. Among other results, \citet{HS19} aim to bound the amount of subsidies assuming that all agent valuations are in $[0,1]$. They conjecture that subsidies of $n-1$ suffice; an even stronger version of the conjecture is proved by \citet{BDN+20}. Complexity results for subsidy minimization are presented by \citet{CI20}. Another study that blends fairness (including envy-freeness) with payments is the work of \citet{CEM17} on distributed allocations of items.

Rent division had attracted attention much before subsidy minimization~\citep*{A95,S99,S83}. A very similar characterization of envy-freeable allocations~\citep*{ADG91,S09} like the one mentioned above, has allowed for a simple solution in~\citep*{GMPZ17} to the problems of maximizing the minimum agent utility and minimizing the disparity between agent utilities. Earlier, the papers~\citep*{ASU04,HRS02,K00} present algorithms for computing envy-free rent divisions without considering optimization criteria. 

\section{Preliminaries}
\label{sec:prelim}
An instance of an allocation problem consists of a set $\agents$ of $n$ agents and a set $\items$ of $m$ items. Agent $i\in \agents$ has valuation $v_i(j)$ for item $j\in \items$ while we use $v_i(S)$ to denote her valuation for the set (or bundle) of items $S$. We assume that valuations are non-negative and additive, i.e., $v_i(S)=\sum_{j\in S}{v_i(j)}$. We remark, though, that our results (including the concept of interim envy-freeness, which we define later in this section) carry over to the more general case of subadditive valuations\footnote{Indeed, Lemmas \ref{lem:ex-ante}, \ref{lem:ex-post}, \ref{lem:ief-implies-prop}, and \ref{lem:max-envy} in Section \ref{sec:ief}, Theorem \ref{thm:poief-upper} in Section \ref{sec:ief-vs-efficiency}, and Theorem \ref{thm:characterization} in Section \ref{sec:payments} still hold for subadditive valuations, while all negative results carry over trivially. Our results in Section \ref{sec:computing} and Subsection \ref{sec:compute-payments} hold for matching instances where only valuations per item matter.}. Furthermore, even though this is rarely required for our positive statements, in our examples we use normalized valuations satisfying $\sum_{j\in \items}{v_i(j)}=1$ for every agent $i\in \agents$.

An allocation $A=(A_1, A_2, ..., A_n)$ of the items in $\items$ to the agents of $\agents$ is simply a partition of the items of $\items$ into $n$ bundles $A_1$, $A_2$, ..., $A_n$, with the understanding that agent $i\in \agents$ gets the bundle of items $A_i$. An allocation $A=(A_1, A_2, ..., A_n)$ is {\em envy-free} (EF) if $v_i(A_i)\geq v_i(A_k)$ for every pair of agents $i$ and $k$. In words, the allocation $A$ is envy-free if no agent prefers the bundle of items that has been allocated
to some other agent to her own. The allocation $A$ is {\em proportional} if $v_i(A_i)\geq \frac{1}{n}v_i(\items)$. An instance may not admit any envy-free or proportional allocation; to see why, consider an instance in which all agents have a positive valuation of $1$ for a single item (and zero value for any other item). It is well-known that, due to additivity, an envy-free allocation is always proportional. 

Envy-freeness is defined accordingly if monetary transfers are allowed. In particular, a pair of an allocation $A$ and a vector $\p$ consisting of a payment $p_i$ to each agent $i\in \agents$ is envy-free with payments if $v_i(A_i)+p_i\geq v_i(A_k)+p_k$ for every pair of agents $i$ and $k$. The term {\em envy-freeable} refers to an allocation that can become envy-free with an appropriate payment vector. Depending on the setting, payments can be restricted to be non-negative (e.g., representing subsidies that are given to the agents~\citep*{BDN+20,CI20,HS19}) or non-positive (e.g., when payments are collected from the agents, like in the rent division problem~\citep*{GMPZ17}).

In addition to their fairness properties, allocations are typically assessed in terms of their efficiency. We say that an allocation $A=(A_1, A_2, ..., A_n)$ is Pareto-efficient if there is no other allocation $A'=(A'_1, A'_2, ..., A'_n)$ with $v_i(A'_i)\geq v_i(A_i)$ for every agent $i\in \agents$, with the inequality being strict for at least one agent of $\agents$. The term social welfare is typically used to assign a cardinal score that characterizes the efficiency of an allocation. Among the several social welfare notions, the utilitarian, egalitarian, and Nash social welfare are the three most prominent. We use the notation $\util(A)$, $\egal(A)$, $\avnash(A)$, and $\lognash(A)$ to refer to the utilitarian, egalitarian, average Nash, and log-Nash social welfare, respectively, of an allocation $A=(A_1, ..., A_n)$; the corresponding efficiency scores are defined as follows:
\begin{align*}
\util(A)=\sum_{i\in\agents}{v_i(A_i)}, \quad \egal(A)=\min_{i\in\agents}{v_i(A_i)}, \quad \avnash(A)=\left(\prod_{i\in\agents}{v_i(A_i)}\right)^{1/n}, \quad \lognash(A)=\sum_{i\in\agents}{\ln{v_i(A_i)}}.
\end{align*}

The {\em price of fairness}, introduced independently in~\citep*{BFT11} and~\citep*{CKKK12}, refers to a class of notions that aim to quantify trade-offs between fairness and efficiency. For example, the price of envy-freeness with respect to the utilitarian social welfare for an allocation instance (that admits at least one envy-free allocation) is the ratio of the optimal utilitarian social welfare of the instance over the utilitarian social welfare of the best envy-free allocation. Different price of fairness notions follow by selecting different fairness concepts and social welfare definitions. 

We are particularly interested in {\em matching allocation instances}, in which the number of agents is equal to the number of items. Our assumptions (e.g., for non-negative valuations) imply that the only reasonably fair (e.g., proportional) allocations should then assign (or match) exactly one item to each agent. We refer to such allocations as {\em matchings} and we typically use the small letter $b$ to denote a matching, instead of the usual notation of $A$ for allocations in general instances. Notice that, an envy-free matching must allocate to each agent her most-valued item. As such, whenever an envy-free allocation exists in a matching instance, it is Pareto-efficient and maximizes the social welfare, according to all definitions of social welfare mentioned above. Hence, the price of envy-freeness is trivially $1$ in this case (with respect to all the social welfare definitions given above).

\subsection*{Random allocations and interim envy-freeness}
We consider random allocations that are produced according to {\em lotteries} (or probability distributions). The lottery $\Q$ over allocations of the items of $\items$ to the agents of $\agents$ is ex-ante envy-free if $\E_{A\sim \Q}[v_i(A_i)]\geq \E_{A\sim \Q}[v_i(A_k)]$ for every pair of agents $i,k\in \agents$. $\Q$ is ex-post envy-free if any allocation it produces with positive probability is envy-free (or, in other words, if all allocations in the {\em support} of $\Q$ are envy-free). 

We now provide the formal definition of the central concept of this paper. We say that a lottery $\Q$ over allocations is {\em interim envy-free} (iEF) if for any pair of agents $i,k\in \agents$ and any possible bundle of items $S$ that agent $i$ can get in a random allocation produced by $\Q$, it holds
\begin{align*}
v_i(S) & \geq \E_{A\sim \Q}[v_i(A_k)|A_i=S].
\end{align*}
The definition of iEF requires that the value agent $i$ has when she gets a bundle $S$ is at least as high as the average value she has for the bundle that agent $k$ gets, conditioned on $i$'s allocation.

We extend the notion of interim envy-freeness to pairs of lotteries over allocations and payments to/from the agents, in an analogous way that recent work has defined envy-freeness with payments. In fact, we differentiate between different payment schemes with respect to whether payments are per agent (A-payments), per bundle (B-payments), or per allocation and agent (C-payments). We similarly extend the notion of price of fairness and envy-freeability to the case of iEF. We postpone providing formal definitions for the corresponding sections that these notions are being examined.

\section{Interim envy-freeness vs. other fairness notions}\label{sec:ief}
In this section we compare interim envy-freeness with other fairness notions, with the aim to identify possible fairness implications. The first implication follows easily by the definitions and has been observed before in more general contexts than ours (e.g., see \citep{Ves04}).

\begin{lemma}\label{lem:ex-ante}
Any iEF lottery $\Q$ is ex-ante envy-free. 
\end{lemma}

\begin{proof}
Indeed, using the definitions of iEF, ex-ante envy-freeness and well-known properties of random variables, we have
\begin{align*}
\E_{A\sim\Q}[v_i(A_i)] &= \sum_{S\subseteq \items}{v_i(S)\cdot \prm_{A\sim\Q}[A_i=S]} \geq \sum_{S\subseteq \items}{\E_{A\sim\Q}[v_i(A_k)|A_i=S]\cdot \prm_{A\sim\Q}[A_i=S]}\\
&= \E_{A\sim\Q}[v_i(A_k)],
\end{align*}
for every pair of agents $i$ and $k$.
\end{proof}

An even simpler observation is that any lottery that deterministically produces an envy-free allocation $A$ is trivially iEF. Indeed, $A_i$ is the only bundle that can be given to agent $i$, who weakly prefers it to the bundle $A_k$ that is allocated to agent $k$. Trivially, $\E_{A\sim\Q}[v_i(A_k)|A_i=S]=v_i(A_k)$ and the iEF condition is identical to the envy-freeness condition $v_i(A_i)\geq v_i(A_k)$. We can slightly extend this argument to obtain the following implication.

\begin{lemma}\label{lem:ex-post}
Any ex-post envy-free lottery $\Q$ is iEF.
\end{lemma}
However, the opposite is not true; we show below that the existence of an iEF allocation does not imply the existence of an EF allocation. This is important as it indicates that the set of iEF allocations is larger than those of EF ones.

\begin{lemma}\label{lem:ief-vs-ef}
There exist allocation instances with an iEF lottery but with no EF allocation.
\end{lemma}

\begin{proof}
Consider the matching allocation instance at the left of Table~\ref{tab:ef-vs-ief-prop} and the lottery $\Q$ which returns matchings $a$-$b$-$c$ and $a$-$c$-$b$, with probability $1/2$ each. As agents $2$ and $3$ have maximum possible value in each of the two matchings produced by $\Q$, EF and, consequently, iEF conditions for them are satisfied. To see that the iEF condition is satisfied for agent $1$, observe that she is allocated item $a$ in both matchings in the support of $\Q$, for which she has a value of $1/3$. Agent $2$ (and, similarly, agent $3$) gets item $b$ with probability $1/2$ and item $c$ with probability $1/2$. Agent $1$'s average value for the item agent $2$ (or agent $3$) gets is $2/3\cdot 1/2 + 0\cdot 1/2 = 1/3$. Hence, the iEF condition for agent $1$ with respect to agent $2$ (and, similarly, for agent $1$ with respect to $3$) is satisfied. The proof that the lottery $\Q$ is iEF is complete.

\begin{table}[ht]
\centering
\caption{The two matching instances that are used in the proofs of Lemmas~\ref{lem:ief-vs-ef} and~\ref{lem:prop-no-ief} to distinguish between iEF, EF, and proportionality. Throughout the paper, we consider several examples with three agents and items $a$, $b$, and $c$. A concise notation like $a$-$c$-$b$ is used to represent the matching in which agents $1$, $2,$ and $3$ get items $a$, $c$, and $b$, respectively. \label{tab:ef-vs-ief-prop}}
{\begin{tabular}{l r}
{\begin{tabular}{c| c cc}
\hline
\phantom{aa}&\phantom{aa}$a$&\phantom{aa}$b$&\phantom{aa}$c$\\
\hline
\phantom{aa}$1$&$1/3$&\phantom{aa}$2/3$&\phantom{aa}$0$\\
\phantom{aa}$2$&\phantom{aa}$0$&\phantom{aa}$1/2$&\phantom{aa}$1/2$\\
\phantom{aa}$3$&\phantom{aa}$0$&\phantom{aa}$1/2$&\phantom{aa}$1/2$\\
\hline
\end{tabular}}
& \quad\quad
{\begin{tabular}{c| cccc}
\hline
\phantom{aa}&\phantom{aa}$a$&\phantom{aa}$b$&\phantom{aa}$c$\\
\hline
\phantom{aa}$1$&\phantom{aa}$1/3$&\phantom{aa}$2/3$&\phantom{aa}$0$\\
\phantom{aa}$2$&\phantom{aa}$0$&\phantom{aa}$2/3$&\phantom{aa}$1/3$\\
\phantom{aa}$3$&\phantom{aa}$1/4$&\phantom{aa}$1/2$&\phantom{aa}$1/4$\\
\hline
\end{tabular}}
\end{tabular}}
\end{table}

In the same example, it can be easily seen that there is no EF allocation. Indeed, as the only agent who has positive value for item $a$ is agent $1$, agent $1$ should get this item and be envious of the agent who gets item $b$. The proof of the lemma is complete.
\end{proof}
Our next lemma relates iEF to proportionality and is used extensively in our proofs.

\begin{lemma}\label{lem:ief-implies-prop}
Any allocation in the support of an iEF lottery is proportional.
\end{lemma}

\begin{proof}
Consider an iEF lottery $\Q$ and any agent $i\in \agents$. By the definition of iEF, we have that for any allocation in the support of $\Q$ where agent $i$ gets the bundle of items $S$, it holds that $v_i(S)\geq \E_{A\sim \Q}[ v_i(A_k)|A_i=S]$ for each other agent $k$. By summing up over all other agents we get
\begin{align*}
(n-1)v_i(S)&\geq \sum_{k\neq i} \E_{A\sim \Q}[ v_i(A_k)|A_i=S] =\E_{A\sim\Q}[\sum_{k\neq i} v_i(A_k)|A_i=S] \geq v_i(\items \setminus S).
\end{align*}
By adding $v_i(S)$ to both sides of the above inequality and rearranging, we get $v_i(S)\geq \frac{1}{n}v_i(\items)$,
implying that any allocation in the support of $\Q$ is proportional.
\end{proof}

However, iEF is a stronger property than proportionality as the next lemma shows.
\begin{lemma}\label{lem:prop-no-ief}
There exist allocation instances with a proportional allocation but with no iEF lottery.
\end{lemma}

\begin{proof}
Consider the instance at the right of Table~\ref{tab:ef-vs-ief-prop}. In this instance, allocation $a$-$c$-$b$ is the only proportional allocation. Hence, by Lemma~\ref{lem:ief-implies-prop}, to show that no iEF lottery exists, it suffices to consider only the (lottery that deterministically returns) allocation $a$-$c$-$b$. In this allocation, agents $1$ and $2$ are envious of agent $3$, contradicting the iEF requirement.
\end{proof}

Our next lemma quantifies the disparity between envy-freeness and interim envy-freeness; the proof exploits Lemma~\ref{lem:prop-no-ief}. 

\begin{lemma}\label{lem:max-envy}
The maximum envy at any allocation in the support of an iEF lottery, when the agent valuations are normalized, can be as high as $1-\frac{2}{n}$ and this is tight.
\end{lemma}
\begin{proof}
Consider the matching instance of Table~\ref{tab:max-envy} with $n$ agents/items.
First, observe that an iEF lottery should always give item $g_1$ to agent $1$. Indeed, any other agent has value $0$ for $g_1$ and could not satisfy the iEF condition if she got it. Hence, in any allocation in the support of an iEF (if one exists), agent 1 has envy $1-\frac{2}{n}$ for the agent who gets item $g_2$. It remains to show that such an iEF lottery does exist. Indeed, any lottery which gives item $g_1$ to agent $1$ and in which the agents $2, 3, ..., n$ are assigned uniformly at random the items $g_2, g_3, ..., g_n$ is iEF. The agents 2, 3, ..., $n$ are clearly not envious of each other or of agent $1$, as each of them gets the maximum possible value. Agent $1$ has expected value $1/n$ for the bundle of each of the remaining agents, which is equal to her value for the item ($g_1$) she gets.

\begin{table}[h!]
\centering
\caption{The instance used in the proof of Lemma~\ref{lem:max-envy}.\label{tab:max-envy}}
{\begin{tabular}{c| c cccc}
\hline
&$g_1$& $g_2$&$g_3$&$\ldots$&$g_n$\\
\hline
$1$&$1/n$&$(n-1)/n$&0&$\ldots$&0\\
$2$&0&$1/(n-1)$&$1/(n-1)$&$\ldots$&$1/(n-1)$\\
$\vdots$&$\vdots$&$\vdots$&$\vdots$&$\vdots$&$\vdots$\\
$n$&0&$1/(n-1)$&$1/(n-1)$&$\ldots$&$1/(n-1)$\\
\hline
\end{tabular}}
{}
\end{table}

To see that $1-2/n$ is the maximum agent envy in any allocation in the support of an iEF lottery, recall that Lemma~\ref{lem:ief-implies-prop} implies that the value of each agent $i$ is at least $1/n$ and, then, the average value $i$ has for the bundle allocated to another agent $k$ cannot exceed $1-1/n$.
\end{proof}
We now compare iEF to two fairness properties for deterministic allocations: min-max fair-share (mFS) and epistemic envy-freeness (EEF). mFS and EEF were considered by ~\citet{BL16} and ~\citet{ABC+18}, respectively, who characterized instances in terms of the properties of the allocations they admit.

Given an allocation instance with a set of agents $\agents$, with valuations $\val$, for the set of items $\items$, let $\mathcal{A}$ be the set of all allocations of the items to the agents. Denote by $\tau_i$ the min-max share of agent $i$, which is defined as
$$\tau_i=\min_{A\in \mathcal{A}}{\max_{j\in \agents}{v_i(A_j)}}.$$
In words, the min-max share $\tau_i$ is the minimum over all allocations of the maximum value agent $i$ has for some bundle. An allocation $A=(A_1, A_2, ..., A_n)$ is mFS (or satisfies the min-max fair-share criterion) if $v_i(A_i)\geq \tau_i$ for every agent $i\in \agents$. mFS should not be confused with maximin share (MMS) fairness, a weaker-than-proportionality fairness property that has received much attention recently~\citep*{KPW18}; see also~\citep*{BL16} for a taxonomy of many fairness properties that we have mentioned here.

An allocation $A=(A_1, A_2, ..., A_n)$ is EEF if for every agent $i\in\agents$ there exists an allocation $B\in \mathcal{A}$ with $B_i=A_i$ such that $v_i(A_i)\geq v_i(B_j)$ for every agent $j\in \agents$. In words, there is a redistribution of the items that agent $i$ does not get to the other agents, so that agent $i$ is not envious.

By this definition, it is clear that an instance that admits an EF allocation also admits an EEF allocation~\citep*{ABC+18}. Any EEF allocation is also mFS~\citep*{ABC+18}, while any mFS allocation is proportional~\citep*{BL16}. These implications are strict~\citep*{ABC+18,BL16}. There are instances that admit proportional allocations but have no mFS allocation, instances with an mFS allocation but no EEF allocation, and instances with an EEF allocation but with no EF allocation.

Where does iEF stand in this hierarchy of properties? So far, we have shown that EF implies iEF which in turn implies proportionality. We have furthermore seen (in Lemmas~\ref{lem:ief-vs-ef} and~\ref{lem:prop-no-ief}) that these implications are strict. Interestingly, iEF is incomparable to both mFS and EEF, as the following two statements show.

\begin{lemma}\label{lem:ief-no-eef-mms}
There exist matching instances with an iEF lottery that do not admit any mFS (and, subsequently, any EEF) allocation.
\end{lemma}

\begin{proof}
The proof uses the instance at the left of Table~\ref{tab:ef-vs-ief-prop} which does have an iEF as explained in the proof of Lemma~\ref{lem:ief-vs-ef}. By definition, in any mFS (and, by the implication shown in~\cite{ABC+18}, any EEF) allocation, the value of each agent should be at least her mFS share. The corresponding shares are $2/3$, $1/2$, and $1/2$ for agents $1$, $2$, and $3$, respectively. We can easily verify that no such allocation exists.
\end{proof}

\begin{lemma}\label{lem:eef-no-ief}
There exist allocation instances with an EEF (and, subsequently, an mFS) allocation but with no iEF lottery.
\end{lemma}

\begin{proof}
Consider the instance shown in Table~\ref{tab:eef-no-ief} and note that this is our only example with a non-matching instance. This is necessary as it can be easily seen that, in matching instances, EF is equivalent to EEF; consequently, EEF implies iEF in these instances.

We first observe that the allocation in which agent $1$ gets item $a$, agent $2$ gets items $b$ and $c$, and agent $3$ gets item $d$ is EEF. Indeed, agent $2$ is clearly non-envious in this distribution. Agent $1$ does not envy the other agents in the redistribution where agent $2$ gets item $b$ and agent $3$ gets items $c$ and $d$. Similarly, agent $3$ does not envy the other agents in the redistribution where agent $2$ gets item $c$ and agent $1$ gets items $a$ and $b$.

\begin{table}[hbt]
\centering
\caption{The instance used in the proof of Lemma~\ref{lem:eef-no-ief}. \label{tab:eef-no-ief}}
{\begin{tabular}{c| cccc}
\hline
\phantom{a}&\phantom{a}$a$&\phantom{a}$b$&\phantom{a}$c$&\phantom{a}$d$\\
\hline
\phantom{a}$1$&\phantom{a}$4/10$&\phantom{a}$3/10$&\phantom{a}$3/10$&\phantom{a}$0$\\
\phantom{a}$2$&\phantom{a}$4/10$&\phantom{a}$3/10$&\phantom{a}$3/10$&\phantom{a}$0$\\
\phantom{a}$3$&\phantom{a}$0$&\phantom{a}$3/10$&\phantom{a}$3/10$&\phantom{a}$4/10$\\
\hline
\end{tabular}}
{}
\end{table}

We now show that no iEF lottery exists. As the proportionality threshold is $1/3$ (observe that valuations are normalized), there are only two proportional allocations: $(\{a\},\{b,c\},\{d\})$ and $(\{b,c\},\{a\},\{d\})$. For any lottery over these two allocations, among agents 1 and 2, the one who gets item $a$ is envious of the other who gets items $b$ and $c$.
\end{proof}
\section{Interim envy-freeness vs. efficiency}\label{sec:ief-vs-efficiency}
We now explore tradeoffs between interim envy-freeness and efficiency. Two well-studied refinements of Pareto-efficiency are relevant for lotteries of allocations: ex-ante and ex-post Pareto-efficiency. A lottery $\Q$ over allocations is ex-ante Pareto-efficient if there exists no other lottery $\Q'$ such that $\E_{A\sim\Q'}[v_i(A_i)]\geq \E_{A\sim\Q}[v_i(A_i)]$ for every agent $i\in\agents$, with the inequality being strict for at least one agent of $\agents$. A lottery is ex-post Pareto-efficient if all allocations in its support are Pareto-efficient. It is well-known that ex-ante Pareto-efficiency implies ex-post Pareto-efficiency. 

For allocation instances with two agents, the allocations in the support of an iEF lottery are envy-free and, thus (as observed in Section~\ref{sec:prelim}), Pareto-efficient. This is due to the fact that any allocation in the support of an iEF lottery is proportional (by Lemma~\ref{lem:ief-implies-prop}) and hence envy-free, since there are only two agents. This implies that an iEF lottery is ex-post and ex-ante Pareto-efficient (it can be easily seen that ex-post and ex-ante Pareto-efficiency coincide for matching allocation instances with two agents; this is not true in general). In the following, we show that this may not be the case in instances with more agents. 

\begin{theorem}\label{thm:po-3}
There exist matching instances with $n\geq 3$ agents in which no iEF lottery is ex-post (and, consequently, ex-ante) Pareto-efficient.
\end{theorem}
\begin{proof}
Consider the matching instance of Table~\ref{tab:no-po}. We will show that the only iEF lottery $\Q$ returns the allocations $a$-$b$-$c$, $a$-$c$-$b$, and $b$-$c$-$a$ equiprobably. Observe that allocation $a$-$c$-$b$ (where all agents have value $1/3$) is Pareto-dominated by both $a$-$b$-$c$ and $b$-$c$-$a$ (in which two agents get value $1/3$ and another agent gets value $2/3$) and, hence, $\Q$ is not ex-post Pareto-efficient. 

\begin{table}[h]
\centering
\caption{The instance in the proof of Theorem~\ref{thm:po-3}. \label{tab:no-po}}
{\begin{tabular}{c| ccc}
\hline
\phantom{aa}&\phantom{aa}$a$&\phantom{aa}$b$&\phantom{aa}$c$\\
\hline
\phantom{aa}$1$&\phantom{aa}$1/3$&\phantom{aa}$2/3$&\phantom{aa}$0$\\
\phantom{aa}$2$&\phantom{aa}$0$&\phantom{aa}$2/3$&\phantom{aa}$1/3$\\
\phantom{aa}$3$&\phantom{aa}$1/3$&\phantom{aa}$1/3$&\phantom{aa}$1/3$\\
\hline
\end{tabular}}
{}
\end{table}

It remains to show that $\Q$ is the only iEF lottery. First, observe that allocations $a$-$b$-$c$, $b$-$c$-$a$, and $a$-$c$-$b$ are the only proportional allocations and, by Lemma~\ref{lem:ief-implies-prop}, the only ones that can be used in the support of another iEF lottery $\Q''$, assuming that one exists. Now, assume that $\Q''$ does not have one of these allocations in its support. We distinguish between three cases:
\begin{itemize}
\item If $a$-$b$-$c$ is not part of the support, then agent $3$ gets item $b$ whenever agent $1$ gets item $a$. The iEF condition is violated for agent $1$ with respect to agent $3$ and bundle $\{a\}$.
\item If $b$-$c$-$a$ is not part of the support, then agent $3$ gets item $b$ whenever agent $2$ gets item $c$. The iEF condition is violated for agent $2$ with respect to agent $3$ and bundle $\{c\}$.
\item If $a$-$c$-$b$ is not part of the support, then agent $2$ gets item $b$ whenever agent $1$ gets item $a$. The iEF condition is violated for agent $1$ with respect to agent $2$ and bundle $\{a\}$.
\end{itemize}
Hence, the three allocations $a$-$b$-$c$, $b$-$c$-$a$, and $a$-$c$-$b$ all appear in the support of $\Q''$. Since $\Q''$ is different than $\Q$, we distinguish between two cases. First, assume that the allocations $a$-$b$-$c$ and $a$-$c$-$b$ have different probabilities in $\Q''$. Then, whenever agent $1$ gets item $a$, some of agents $2$ and $3$, call her $i$, gets item $b$ with probability strictly higher than $1/2$. Hence, the expected value of agent $1$ for the bundle of agent $i$ is strictly higher than $1/3$ and the iEF condition is violated for the pair of agents $1$, $i$ and bundle $\{a\}$. The other case in which the allocations $b$-$c$-$a$ and $a$-$c$-$b$ have different probabilities in $\Q''$ is symmetric. We conclude that $\Q''$ is not iEF.

We complete the proof by showing that the lottery $\Q$ is indeed iEF. Agent $3$ is never envious. The iEF condition for agent $1$ (or agent $2$) clearly holds when she gets item $b$ (since her value for item $b$ is the highest). When agent $1$ gets item $a$ (respectively, when agent $2$ gets item $c$), she has a value of $1/3$. Then, each of the other two agents $2$ and $3$ (respectively, $1$ and $3$) gets equiprobably items $b$ and $c$ (respectively, $a$ and $b$). The expected value agent $1$ (respectively, agent $2$) has for the item allocated to agent $2$ or agent $3$ (respectively to agent $1$ or agent $3$) is equal to $1/3$, and the iEF condition holds with equality.
\end{proof}

To assess the impact of fairness in random allocations to social welfare, we need to extend the price of fairness definition to lotteries. We do so implicitly here by defining the price of iEF (one can similarly define, e.g., the price of ex-ante envy-freeness). We say that the price of iEF with respect to a social welfare measure is the worst-case ratio over all allocation instances with at least one iEF lottery, of the optimal social welfare in the instance over the expected social welfare of the best iEF lottery (where ``best'' is defined with respect to this social welfare measure).

In our next theorems, we bound the price of iEF with respect to different social welfare notions. In the proof of our upper bounds, we consider normalized valuations. This is a typical assumption in the related literature as well, e.g., see~\citep*{CKKK12}. This assumption is not necessary for average Nash social welfare.

\begin{theorem}\label{thm:poief-upper}
The price of iEF with respect to the utilitarian, egalitarian, and average Nash social welfare is at most $n$, when the agent valuations are normalized.
\end{theorem}

\begin{proof}
The proof follows by Lemma~\ref{lem:ief-implies-prop}, which implies that the valuation of each agent in any allocation in the support of an iEF lottery is at least $1/n$. Then, the utilitarian, egalitarian, and average Nash social welfare is at least $1$, $1/n$, and $1/n$, respectively, while the corresponding optimal values are at most $n$, $1$, and $1$, respectively.
\end{proof}

The next two statements indicate that our price of iEF upper bounds with respect to utilitarian and egalitarian social welfare are asymptotically tight.

\begin{theorem}\label{thm:poief-util-lower}
The price of iEF with respect to the utilitarian social welfare is at least $\Omega(n)$.
\end{theorem}

\begin{proof}
Let $\epsilon>0$ be negligibly small and $k\geq 2$ be an integer. We use the following matching instance with $n=2k$ agents and items. For $i=1, 2, ..., k$, agent $i$ has value $\frac{k}{k+1}$ for item $i$, value $\frac{1}{k+1}$ for item $i+k$, and value $0$ for any other item. For $i=k+1, k+2, ..., 2k$, agent $i$ has value $\frac{1}{2k}+\epsilon$ for items $1, 2, ..., k$ and value $\frac{1}{2k}-\epsilon$ for items $k+1, ..., 2k$.

An optimal allocation has utilitarian social welfare (at least) $\frac{k^2}{k+1}+\frac{1}{2}-k\epsilon$. To see why, consider the allocation in which agent $i$ gets item $i$ for $i=1, 2, ..., 2k$. We now claim that no iEF lottery $\Q$ over allocations has welfare higher than $\frac{k}{k+1}+\frac{1}{2}+k\epsilon$. The lower bound on the price of iEF will follow by the relation between $n$ and $k$ (and by taking $\epsilon$ to be sufficiently small). Indeed, by Lemma~\ref{lem:ief-implies-prop}, the support of $\Q$ should consist of allocations in which agents $k+1, k+2, ..., 2k$ get items $1, 2, ..., k$ for a total value of $\frac{1}{2}+k\epsilon$. Then, the maximum value each of the agents $1, 2, ..., k$ gets from the items $k+1, ..., 2k$ is $\frac{1}{k+1}$.

It remains to present such a lottery $\Q$. It suffices to assign item $i+k$ to agent $i$ for $i=1, 2, ..., k$ and assign uniformly at random the items $1, 2, ..., k$ to the agents $k+1, ..., 2k$. Clearly, agents $k+1, ..., 2k$ are not envious. For $i=1, ..., k$, agent $i$ has value $\frac{1}{k+1}$. Her expected value for the item of another agent $\ell$ is $0$ if $\ell$ is one of the $k$ first agents besides $i$ and $\frac{1}{k+1}$ if $\ell$ is one of the $k$ last agents. Notice that in the latter case, agent $\ell$ gets item $i$ (for which agent $i$ has value $\frac{k}{k+1}$) with probability $1/k$, while she gets items for which agent $i$ has no value otherwise.
\end{proof}

\begin{theorem}\label{thm:poief-egal-lower}
The price of iEF with respect to the egalitarian social welfare is at least $\Omega(n)$.
\end{theorem}
\begin{proof}
Consider the following matching instance with $n$ agents and items. Agent $1$ has value $1/3$ for item 1 and value $2/3$ for item 2. For $i=2, 3, ..., n-2$, agent $i$ has value $1/3$ for each of the items $i$, $i+1$, and $n$. Agent $n-1$ has value $1/2$ for each of the last two items, while agent $n$ has value $\frac{1}{n-1}$ for item $1$ and $1-\frac{1}{n-1}$ for item $n$. All other agent valuations are $0$.

An optimal allocation has egalitarian social welfare at least $1/3$. To see why, consider the allocation where agent $i$ gets item $i$, for $i=1, 2, ..., n$. Furthermore, we claim that any iEF lottery $\Q$ over allocations has egalitarian social welfare $\frac{1}{n-1}$. Indeed, by Lemma \ref{lem:ief-implies-prop}, any allocation in the support of $\Q$ has agent $1$ getting either item $1$ or item $2$. In the former case, agent $2$ must get item $2$ and, hence, agent $1$ envies agent $2$. So, agent $1$ gets item $2$ and, again by Lemma \ref{lem:ief-implies-prop}, item $1$ must be assigned to agent $n$.

Consider now the following lottery. Agent $1$ gets item $2$, agent $n$ gets item $1$, while the remaining items are assigned in the following way. Let $\ell$ be selected equiprobably from $\{2, \dots, n-1\}$. Then, agent $\ell$ gets item $n$, each agent $i$, for $i=2, ..., \ell-1$, gets item $i+1$ and each agent $i$, for $i=\ell+1, ..., n-1$, gets item $i$. Note that all agents, besides agent $n$, get an item of maximum value to them. Agent $n$ obtains value $\frac{1}{n-1}$ from item $1$; she does not envy another agent, as, apart from item $1$, she only values item $n$ positively and any other agent gets item $n$ with probability at most $\frac{1}{n-2}$.
\end{proof}

A very similar proof to the one of Theorem~\ref{thm:poief-util-lower} yields our best (albeit not known to be tight) lower bound for the price of iEF with respect to the average Nash social welfare.

\begin{theorem}\label{thm:poief-nash-lower}
The price of iEF with respect to the average Nash social welfare is at least $\Omega\left(\sqrt{n}\right )$.
\end{theorem}
\begin{proof}
The proof uses the same instance with the one in the proof of Theorem~\ref{thm:poief-util-lower} and a very similar reasoning. By considering again the allocation in which agent $i$ gets item $i$ for $i=1, 2, ..., 2k$, we get that the optimal allocation has average Nash social welfare at least $\left(\left(\frac{k}{k+1}\right)^k\left(\frac{1}{2k}-\epsilon\right)^k\right)^{\frac{1}{2k}}$.

As we argued in the proof of Theorem~\ref{thm:poief-util-lower}, the support of any iEF lottery $\Q$ consists of allocations in which agents $k+1$, $k+2$, ..., $2k$ get distinct items among items $1$, $2$, ..., $k$; each of these agents has value $\frac{1}{2k}+\epsilon$. Then, for $i=1, 2, ..., k$, agent $i$ gets item $k+i$ and has value $\frac{1}{k+1}$. The average Nash social welfare of such allocations (which do exist as we have also shown) is $\left(\left(\frac{1}{k+1}\right)^k\left(\frac{1}{2k}+\epsilon\right)^k\right)^{\frac{1}{2k}}$.

By setting $\epsilon=\frac{1}{6k}$, we obtain that the price of iEF with respect to the average Nash social welfare is (at least)
\begin{align*}
\left(\frac{\left(\frac{k}{k+1}\right)^k\left(\frac{1}{2k}-\epsilon\right)^k}{\left(\frac{1}{k+1}\right)^k\left(\frac{1}{2k}+\epsilon\right)^k}\right)^{\frac{1}{2k}} &=\sqrt{\frac{k}{2}}=\frac{\sqrt{n}}{2}.
\end{align*}
The theorem follows.
\end{proof}
\section{Computing efficient interim envy-free allocations}
\label{sec:computing}

We devote this section to proving Theorem \ref{thm:compute-ief} below, and show how we can compute an iEF lottery of maximum expected social welfare efficiently in the case of matching instances. We remark that, in the case of non-matching instances (where the number of items is larger than the number of agents), deciding whether an iEF lottery exists is an NP-hard problem. To see why, consider the case of two agents and recall that (i) an iEF lottery is a lottery over proportional allocations (Lemma \ref{lem:ief-implies-prop}), and (ii) a proportional allocation is envy-free and (trivially) iEF. Hence, the existence of an iEF lottery is equivalent to the existence of a proportional/EF allocation in the case of two agents. Now, if we consider instances with identical valuations (in which both agents have value $v(g)$ for item $g$), deciding whether an iEF lottery exists is equivalent to deciding {\sc Partition}, a well-known NP-hard problem.

\begin{theorem}\label{thm:compute-ief}
For matching instances, an iEF lottery of maximum expected utilitarian, egalitarian, or log-Nash social welfare (if one exists) can be computed in polynomial time in terms of the number of agents.
\end{theorem}

Our algorithms use linear programming. Let $\M$ be the set of all possible perfect matchings between the agents in $\agents$ and the items in $\items$ (more formally, $\M$ is the set of all perfect matchings in the complete bipartite graph $G=(\agents,\items,\agents\times\items)$). For agent $i\in \agents$ and item $j\in \items$, denote by $\M_{ij}$ the set of matchings from $\M$ in which item $j$ is assigned to agent $i$. Also, for a matching $b\in \M$ and an agent $k\in\agents$, $b(k)$ denotes the item of $\items$ to which agent $k$ is matched in $b$. Then, an iEF lottery can be computed as the solution to the following linear program.
\begin{equation}\label{LP1}
\begin{array}{ll@{}ll}
\text{maximize} & \displaystyle\sum\limits_{b\in \M} &x(b)\cdot \sw(b) & \\
\text{subject to}& \displaystyle\sum\limits_{b\in \M_{ij}}  &x(b)\cdot (v_i(j)-v_i(b(k))) \geq 0, & i\in \agents, j\in \items, k\in \agents\setminus\{i\}\\
 & \displaystyle\sum\limits_{b\in \M}  &x(b) =1 &\\
 &                                                &x(b) \geq 0, & b\in \M
\end{array}
\end{equation}
The variables of the linear program are the probabilities $x(b)$, for every matching $b\in \mathcal{M}$, with which the lottery produces matching $b$. Together with the non-negativity constraints on $\x$, the second constraint $\sum_{b\in \M}{x(b)}=1$ requires that the vector of probabilities $\x=(x(b))_{b\in \M}$ defines a lottery over all matchings of $\M$. The notation $\sw(b)$ is used here to refer generally to the social welfare of matching $b$. We will specifically replace $\sw$ by $\util$, $\egal$, and $\lognash$ later. The objective of the linear program is to maximize the expected social welfare $\E_{b\sim \x}[\sw(b)]$ or, equivalently, the quantity $\sum_{b\in \M}{x(b)\cdot \sw(b)}$.

The first set of constraints represent the iEF conditions. Indeed, the constraint is clearly true for every agent $i\in\agents$ and item $j\in\items$ that is never assigned to agent $i$ under $\x$ (i.e., when $\Pr_{b\sim\x}[b(i)=j]=0$). We will also show that this is the case when $\Pr_{b\sim\x}[b(i)=j]>0$ as well. For agent $k\in \agents\setminus\{i\}$, interim envy-freeness requires that
\begin{align}\label{eq:ief-condition}
    v_i(j) &\geq \E_{b\sim \x}[v_i(b(k))|b(i)=j].
\end{align}
By multiplying the left-hand-side of (\ref{eq:ief-condition}) with $\prm_{b\sim\x}[b(i)=j]$, we get
\begin{align*}
    v_i(j)\cdot \prm_{b\sim\x}[b(i)=j] &= \sum_{b\in \M_{ij}}{x(b)\cdot v_i(j)},
\end{align*}
and by doing the same with the right-hand-side, we have
\begin{align*}
    \E_{b\sim \x}\left[v_i(b(k))|b(i)=j]\cdot \prm_{b\sim\x}[b(i)=j\right] &= \sum_{b\in\M_{ij}}{x(b)\cdot v_i(b(k))}.
\end{align*}
Hence, inequality (\ref{eq:ief-condition}) is equivalent to the first constraint of the linear program (\ref{LP1}).

The linear program (\ref{LP1}) has exponentially many variables, i.e., one variable for each of the $n!$ different matchings of $\M$. To solve it efficiently, we will resort to its dual linear program
\begin{equation}\label{dual-LP1}
\begin{array}{llll}
\text{maximize}  & z & & \\
\text{subject to}& z+\sum\limits_{\substack{(i,j)\in b\\k\in \agents\setminus \{i\}}}{(v_i(j)-v_i(b(k)))\cdot y(i,j,k)} +\sw(b)\leq 0, & b\in \M\\
 & y(i,j,k) \geq 0,  i\in \agents, j\in \items, k\in\agents\setminus\{i\},
\end{array}
\end{equation}
where variable $y(i,j,k)$ corresponds to the iEF constraint between agents $i$ and $k$ when $i$ is allocated item $j$.
The dual linear program (\ref{dual-LP1}) has polynomially many variables and exponentially many constraints. Fortunately, we will be able to solve it using the ellipsoid method~\citep*{GLS88,S86}. To do so, all we need is a polynomial-time separation oracle, which takes as input values for the dual variables $z$ and $y(i,j,k)$ for all triplets $(i,j,k)$ consisting of agent $i\in \agents$, item $j\in\items$, and agent $k\in \agents\setminus\{i\}$, and either computes a matching $b^*$ for which a particular constraint is violated, or correctly concludes that no constraint of the dual linear program (\ref{dual-LP1}) is violated.

Let us briefly remind the reader how solving the dual linear program using the ellipsoid method can give us an efficient solution to the primal linear program as well; a more detailed discussion can be found in~\cite{GLS88,S86}. To solve the dual linear program, the ellipsoid method will make only polynomially many calls to the separation oracle. This is due to the fact that, among the exponentially many constraints, the ones that really constrain the variables of the dual linear program are very few; the rest are just redundant. Then, after having kept track of the execution of the ellipsoid method on the dual linear program, the primal linear program can be simplified by setting implicitly to $0$ all variables that correspond to dual constraints that were not returned as violated ones by the calls of the separation oracle during the execution of the ellipsoid method. As a final step, the solution of the simplified primal linear program (which is now of polynomial size) will give us the solution $\x$; this will have only polynomially-many matchings in its support.

In the rest of this section, we will show how to design such separation oracles for the dual linear program (\ref{dual-LP1}) when we use the utilitarian, egalitarian, or log-Nash definition of the social welfare. Our separation oracles essentially solve instances of a novel variation of the maximum bipartite matching problem. We believe that this can be of independent interest, with applications in many different contexts.

\subsection*{The maximum edge-pair-weighted perfect bipartite matching}
Instances of the maximum edge-pair-weighted perfect bipartite matching problem (or, 2EBM, for short) consist of the complete bipartite graph $G=(\agents,\items,\agents\times\items)$, with $|\agents| = |\items|$, and a weighting function $\psi$ that assigns weight $\psi(e_1,e_2)$ to every ordered pair of non-adjacent edges $e_1$ and $e_2$ from $\agents\times\items$. The objective is to compute a perfect matching $b\in\M$ so that the total weight over all edge-pairs of $b$, denoted by
\begin{align}\label{2ebm-objective}
\Psi(b) &= \sum_{(i,j)\in b}{\sum_{\substack{(k,\ell)\in b:\\k\not=i}}{\psi\left((i,j),(k,\ell)\right)}},
\end{align}
is maximized. We remark that the problem of computing a perfect matching of maximum total edge weight in an edge-weighted complete bipartite graph with edge weight $w(e)$ for each edge $e$, corresponds to 2EBM with edge-pair weights defined as
$\psi(e,e')=\frac{w(e)}{n-1}$ for every ordered pair of edges $e$ and $e'$. The separation oracle used for the solution of the dual linear program (\ref{dual-LP1}) will use appropriately defined weights $\psi$ that will depend on the variables that appear at the left-hand side of the constraint of (\ref{dual-LP1}) in a way that will be specified in the following Subsections \ref{sec:util}, \ref{sec:egal}, and \ref{sec:log-nash}. 

Let $\X$ be the set of quadruples $(i,j,k,\ell)$ where $i,k\in \agents$ and $j,\ell \in \items$, with $i\not=k$ and $j\not=\ell$. We can view such a quadruple as the ordered pair of edges $(i,j)$ and $(k,\ell)$ in the input graph $G$. Essentially, the quadruples of $\X$ correspond to all possible ordered pairs of different edges in the input graph. To compute a perfect matching $b\in \M$ of maximum total edge-pair weight, we will use the following integer linear program with $\Theta(n^4)$ variables and constraints (where $n$ is the number of agents and items). We remark that, from now on, we simplify the notation $\psi((i,j),(k,\ell))$ and use $\psi(i,j,k,\ell)$ instead.
\begin{equation}\label{lp:max-epw-perfect-matching}
\begin{array}{llll}
\text{maximize}  & \sum\limits_{(i,j,k,\ell)\in \X}{t(i,j,k,\ell)\cdot \psi(i,j,k,\ell)} & & \\
\text{subject to} & \sum\limits_{\substack{j,\ell\in \items:\\(i,j,k,\ell)\in \X}}{t(i,j,k,\ell)} = 1, & i\in \agents, k\in \agents\setminus\{i\}\\
& \sum\limits_{\substack{i,k\in \agents:\\(i,j,k,\ell)\in \X}}{t(i,j,k,\ell)} = 1, & j\in \items, \ell\in\items\setminus\{j\}\\
& t(i,j,k,\ell) = t(k,\ell,i,j), & (i,j,k,\ell)\in \X\\
& t(i,j,k,\ell) \in \{0,1\}, & (i,j,k,\ell)\in \X
\end{array}
\end{equation}
For a quadruple $(i,j,k,\ell)\in \X$, the variable $t(i,j,k,\ell)$ indicates whether both edges $(i,j)$ and $(k,\ell)$ belong to the perfect matching ($t(i,j,k,\ell)=1$) or not ($t(i,j,k,\ell)=0$). The first constraint indicates that among all edge pairs with endpoints at agent nodes $i$ and $k$, exactly one has both its edges in the matching. Similarly, the second constraint indicates that among all edge pairs with endpoints at item nodes $j$ and $\ell$, exactly one has both its edges in the matching. The third constraint ensures symmetry of the variables so that they are consistent to our interpretation. 

We relax the integrality constraint of (\ref{lp:max-epw-perfect-matching}) and replace it by
\begin{align}\label{eq:lp-non-negative-variables}
t(i,j,k,\ell) \geq 0, & \quad\quad\quad\quad (i,j,k,\ell)\in \X.
\end{align}
Then, we compute an extreme solution of the resulting linear program (e.g., again, using the ellipsoid method~\citep*{GLS88,S86}). We claim (in Lemma~\ref{lem:extreme}) that this solution is integral, i.e., all variables have values either $0$ or $1$, and are, hence, solutions to the integer linear program (\ref{lp:max-epw-perfect-matching}) and, consequently, to our maximum edge-pair-weighted perfect bipartite matching problem.

To prove this, we can view the solution $\tbold$ of the relaxation of the linear program (\ref{lp:max-epw-perfect-matching}) as a square matrix $T$. In this matrix, each row corresponds to a pair of different agents and each column to a pair of different items. Then, the first and the second set of constraints indicate that $T$ is {\em doubly stochastic}. The additional symmetry constraint of the linear program (\ref{lp:max-epw-perfect-matching}) prevents us from using the famous Birkhoff-von Neumann theorem~\citep{B46}, which states that any doubly stochastic matrix is a convex combination of permutation matrices (i.e., square binary matrices with exactly one $1$ at each row and each column). 
Fortunately, we can use an extension due to~\citet{C77} which proves a similar result for
centro-symmetric matrices.

\begin{defn}
An $N\times N$ matrix $T=(T_{u,v})_{u,v\in[N]}$ is called centro-symmetric if it satisfies $T_{u,v}=T_{N+1-u,N+1-v}$ for all $u,v\in [N]$.
\end{defn}

\begin{theorem}[\cite{C77}]\label{thm:cruse}
If $N$ is even, then any $N\times N$ centro-symmetric doubly stochastic matrix is the convex combination of centro-symmetric permutation matrices.
\end{theorem}

We use Theorem~\ref{thm:cruse} in the proof of the next lemma which shows that the extreme solutions of the relaxation of the linear program (\ref{lp:max-epw-perfect-matching}) are integral and, hence, correspond to perfect matchings.

\begin{lemma}\label{lem:extreme}
Any extreme solution of the relaxation of the linear program (\ref{lp:max-epw-perfect-matching}) is integral.
\end{lemma}

\begin{proof}
We will define an alternative representation of a feasible solution $\tbold$ of the relaxation of the linear program (\ref{lp:max-epw-perfect-matching}) as a doubly stochastic matrix $T$ with $N=n(n-1)$ rows and columns. To do so, we will use a particular mapping of each pair of different agents (respectively, of each pair of different items) to particular rows (respectively, columns) of the matrix $T$. This particular mapping will allow us to argue that the matrix $T$ is centro-symmetric. As $N$ is even, Theorem~\ref{thm:cruse} will give us that $T$ is a convex combination of centro-symmetric permutation matrices, which correspond to integral solutions.

As both sets $\agents$ and $\items$ contain $n$ elements each, we may view them as integers from $[n]$. We define the bijection $\pi$ from ordered pairs of different integers from $[n]$ to integers of $[N]$ as follows. For every ordered pair $(i,k)$ of different integers from $[n]$, let
\begin{align*}
\pi(i,k) &= \sum_{h=1}^{i-1}{(n-h)}+k-i
\end{align*}
if $i<k$, and
\begin{align*}
\pi(i,k) &=n(n-1)+1-\pi(k,i)
\end{align*}
otherwise. By this definition, we have
\begin{align}\label{eq:centro}
\pi(i,k)+\pi(k,i) &= n(n-1)+1.
\end{align}
Note that, for $i=1, 2, ..., n-1$ and $k=i+1, ..., n$, $\pi(i,k)$ takes all distinct integer values from $1$ to $n(n-1)/2$. Then, for the remaining pairs $(i,k)$ with $i=2, ..., n$ and $k=1, ..., i-1$, $\pi(i,k)$ takes all distinct integer values from $n(n-1)$ down to $n(n-1)/2+1$. Hence, since $N=n(n-1)$, each distinct ordered pair of different integers from $[n]$ is mapped to a different integer of $[N]$ under $\pi$. Hence, $\pi$ is indeed a bijection. Now, for every quadruple $(i,j,k,\ell)\in \X$, we store the value of $t(i,j,k,\ell)$ in the entry $T_{\pi(i,k),\pi(j,\ell)}$ of matrix $T$. By the properties of $\pi$, the matrix $T$ is well-defined.

We will complete the proof by showing that $T$ is centro-symmetric. Indeed, let $u$ and $v$ be any integers in $[N]$ and assume that $u=\pi(i,k)$ and $v=\pi(j,\ell)$ for pairs of distinct integers $(i,k)$ and $(j,\ell)$. We have
\begin{align*}
T_{u,v} &= T_{\pi(i,k),\pi(j,\ell)} = t(i,j,k,\ell) = t(k,\ell,i,j)\\
&= T_{\pi(k,i),\pi(\ell,j)}=T_{N+1-\pi(i,k), N+1-\pi(j,\ell)} = T_{N+1-u,N+1-v},
\end{align*}
i.e., $T$ is indeed centro-symmetric. The first and sixth equalities follow by the definition of $u$ and $v$. The second and fourth equalities follow by the definition of the entries of matrix $T$. The third equality is the symmetry constraint of linear program (\ref{lp:max-epw-perfect-matching}). The fifth equality follows by (\ref{eq:centro}).
\end{proof}

Hence, the execution of the ellipsoid algorithm on the relaxation of the linear program (\ref{lp:max-epw-perfect-matching}) will return an integral solution that corresponds to a solution of 2EBM. The next statement summarizes the above discussion.

\begin{theorem}
2EBM can be solved in polynomial time.
\end{theorem}

We are ready to show how solutions to appropriately defined instances of 2EBM can be used as separation oracles for solving the linear program (\ref{dual-LP1}) when $\sw$ is the utilitarian (Section~\ref{sec:util}), egalitarian (Section~\ref{sec:egal}), and log-Nash (Section~\ref{sec:log-nash}) social welfare.

\subsection{Utilitarian social welfare}\label{sec:util}
By the definition of the utilitarian social welfare, we have
\begin{align*}
    \util(b) &=\sum_{(i,j)\in b}{v_i(j)}= \sum_{(i,j)\in b}\sum_{\substack{(k,\ell)\in b:\\k\not=i}}{\frac{v_i(j)+v_k(\ell)}{2(n-1)}},
\end{align*}
and, using $\sw(b)=\util(b)$, the constraint of the dual linear program (\ref{dual-LP1}) corresponding to a matching $b\in \M$ is equivalent to
\begin{align}\label{eq:constr-sep-U}
    \sum_{(i,j)\in b}\sum_{\substack{(k,\ell)\in b:\\k\not=i}}{\left((v_i(j)-v_i(\ell))\cdot y(i,j,k)+\frac{v_i(j)+v_k(\ell)}{2(n-1)}+\frac{z}{n(n-1)}\right)} \leq 0.
\end{align}
So, consider the instance of 2EBM with edge weights defined as
\begin{align*}
    \psi(i,j,k,\ell) &= (v_i(j)-v_i(\ell))\cdot y(i,j,k)+\frac{v_i(j)+v_k(\ell)}{2(n-1)}+\frac{z}{n(n-1)}.
\end{align*}
Then, for a matching $b\in \M$, we have that the objective function of 2EBM, $\Psi(b)$, shown in (\ref{2ebm-objective}), is equal to the left-hand-side of inequality (\ref{eq:constr-sep-U}) and, consequently, to the left-hand-side of the constraint of the dual linear program (\ref{dual-LP1}), when the utilitarian definition of the social welfare is used.

Now, the separation oracle for the dual linear program (\ref{dual-LP1}) works as follows. It solves the instance of 2EBM just described and computes a matching $b^*\in \M$ that maximizes the quantity $\Psi(b)$, i.e., $b^*\in \argmax_{b\in \M}{\Psi(b)}$. If $\Psi(b^*)>0$, the constraint corresponding to the matching $b^*$ in the dual linear program (\ref{dual-LP1}) is returned as a violating constraint. Otherwise, it must be $\Psi(b)\leq 0$ for every matching $b\in \M$ and the separation oracle correctly returns that no such violating constraint exists.

\subsection{Egalitarian social welfare}\label{sec:egal}
Let $L$ denote the different values the valuations $v_i(j)$ of an agent $i$ for item $j$ can get, i.e., $L=\{v_i(j):i\in \agents, j\in \items\}$. For $e\in L$, denote by $\M_e$ the set of perfect matchings so that for any agent $i$ that is assigned to item $j$, it holds that $v_i(j)\geq e$. Observe that the perfect matching $b\in \M$ belongs to set $\M_e$ for every $e\leq \egal(b)$. Then, the constraints of the dual linear program (\ref{dual-LP1}) for the egalitarian definition of the social welfare are captured by the following set of constraints:
\begin{align}\label{eq:constr-sep-E}
    \sum_{(i,j)\in b}\sum_{\substack{(k,\ell)\in b:\\k\not=i}}{\left((v_i(j)-v_i(\ell))\cdot y(i,j,k)+\frac{e+z}{n(n-1)}\right)} \leq 0, \quad\quad e\in L, b\in \M_e
\end{align}
Indeed, for every matching $b\in \M$, the set of constraints (\ref{eq:constr-sep-E}) contains the constraint corresponding to $b$ in the dual linear program (\ref{dual-LP1}) with $\sw=\egal$ and, possibly, the redundant constraints
\begin{align*}
z+\sum\limits_{\substack{(i,j)\in b\\k\in \agents\setminus \{i\}}}{(v_i(j)-v_i(b(k)))\cdot y(i,j,k)} +e \leq 0,
\end{align*}
for $e\in L$ with $e<\egal(b)$ (if any). So, to design the separation oracle for the dual linear program (\ref{dual-LP1}), it suffices to design a separation oracle for the set of constraints (\ref{eq:constr-sep-E}), for each of the $\bigO(n^2)$ different values of $e\in L$. We now show how to do so.

For $e\in L$, let $\X_e$ be the subset of $\X$ such that $v_i(j)\geq e$ and $v_k(\ell)\geq e$. Essentially, the quadruples of $\X_e$ correspond to all possible (ordered) pairs of different edges in a perfect matching of $\M_e$. Now, for every $e\in L$, consider the instance of 2EBM with weights
\begin{align*}
    \psi(i,j,k,\ell) &= (v_i(j)-v_i(\ell))\cdot y(i,j,k)+\frac{e+z}{n(n-1)}
\end{align*}
for quadruple $(i,j,k,\ell)\in \X_e$. Then, for a matching $b\in \M_e$, the objective function of 2EBM, $\Psi(b)$, shown in (\ref{2ebm-objective}), is equal to the left-hand-side of inequality (\ref{eq:constr-sep-E}). Now, for each $e\in L$, the separation oracle computes the matching $b^*_e$ that maximizes the quantity $\Psi(b)$ among all matchings of $\M_e$. A violating constraint (corresponding to matching $b^*_e$) is then found if $\Psi(b^*_e)>0$ for some $e\in L$. 
Otherwise, the separation oracle concludes that no constraint is violated.

\subsection{Log-Nash social welfare}\label{sec:log-nash}
Observe that the definition of the log-Nash social welfare implies
\begin{align*}
    \lognash(b) &=\sum_{(i,j)\in b}{\ln{v_i(j)}}= \sum_{(i,j)\in b}\sum_{\substack{(k,\ell)\in b:\\k\not=i}}{\frac{\ln{v_i(j)}+\ln{v_k(\ell)}}{2(n-1)}}.
\end{align*}
Hence, the only modification that is required in the approach we followed for the utilitarian social welfare is to replace the term $\frac{{v_i(j)}+{v_k(\ell)}}{2(n-1)}$ by $\frac{\ln{v_i(j)}+\ln{v_k(\ell)}}{2(n-1)}$ in the definition of $\psi$.

\section{Interim envy-free allocations with payments}
\label{sec:payments}
In this section, we extend the notion of interim envy-freeness by accompanying lotteries over allocations with payments to/from the agents. In this case, the definition of iEF uses both the value an agent has for item bundles and the payment she receives or contributes. We distinguish between three different types of payments. A vector of {agent-dependent payments} or {\em A-payments} consists of a payment $p_i$ for each agent $i\in \agents$. More refined payments are defined using additional information for an allocation instance. We say that the agents receive {\em bundle-dependent payments} or {\em B-payments} when each agent is associated with a payment of $p(S)$ when she receives the bundle of items $S$. Finally, we say that the agents get {\em allocation-dependent payments} or C-payments when each agent $i$ is associated with payment $p_i(A)$ in allocation $A$.

We now extend the notion of interim envy-freeness to pairs of lotteries and payments by distinguishing between the three payment types.

\begin{defn}\label{defn:A-payments}
We say that a pair of a lottery $\Q$ and a vector of A-payments $\p$ is iEF if for every pair of agents $i, k\in \agents$ and every bundle of items $S$ agent $i$ can get under $\Q$, it holds $v_i(S) + p_i \geq \E_{A\sim \Q}[v_i(A_k)|A_i=S]+p_k$.
\end{defn}

\begin{defn}\label{defn:B-payments}
We say that a pair of a lottery $\Q$ and a vector of B-payments $\p$ per bundle of items is iEF if for every pair of agents $i, k\in \agents$ and every bundle of items $S$ agent $i$ can get under $\Q$, it holds $v_i(S) + p(S) \geq \E_{A\sim \Q}[v_i(A_k)+p(A_k)|A_i=S]$.
\end{defn}

For C-payments, we give a more general definition that allows for marginal violations of iEF. The notion of $\epsilon$-iEF will be useful later in Section~\ref{sec:compute-payments}.
\begin{defn}\label{defn:C-payments}
Let $\epsilon\geq 0$. We say that a pair of a lottery $\Q$ and a vector of C-payments $\p$ per agent and allocation is $\epsilon$-iEF if for every pair of agents $i, k\in \agents$ and every bundle of items $S$ agent $i$ can get under $\Q$, it holds $v_i(S) + \E_{A\sim \Q}[p_i(A)|A_i=S] \geq \E_{A\sim \Q}[v_i(A_k)+p_k(A)|A_i=S]-\epsilon$.
\end{defn}
The term iEF with C-payments is used alternatively to $0$-iEF. We remark that the payments are added to the value agents have for bundles in the above definitions. So, in general, the payments are assumed to be received by the agents. To represent payments that are contributed by the agents, we can allow negative entries in the payment vectors. We also remark that the definitions refer to general allocation instances. Indeed, the result we present in Section~\ref{sec:char} applies to general instances. Then, in Sections~\ref{sec:opt} and~\ref{sec:compute-payments}, we restrict our attention to matching instances and adapt the definitions accordingly.

Our technical contribution regarding iEF allocations with payments is many-fold. First, we characterize in Section~\ref{sec:char} those lotteries $\Q$ that can be complemented with vectors of A-payments $\p$, so that the pair ($\Q,\p$) is iEF with A-payments. There, our focus is on the existence of payments, with no additional restrictions on them. In Sections~\ref{sec:opt} and~\ref{sec:compute-payments}, we specifically consider two particular optimization problems that involve iEF allocations with payments; we define these problems in the following.

In both optimization problems, we are given an allocation instance and the objective is to compute a lottery $\Q$ over allocations and a payment vector $\p$ so that the pair ($\Q,\p$) is iEF with payments. In {\em subsidy minimization}, the payments are subsidies given to the agents by an external authority. So, the corresponding entries in the payment vector $\p$ are constrained to be non-negative. The goal of subsidy minimization is to find an iEF allocation and accompanying payments, such that the total expected amount of subsidies is minimized; the objective is equal to $\sum_{i\in \agents}{p_i}$, $\E_{A\sim\Q}[\sum_{i\in\agents}{p(A_i)}]$, or $\E_{A\sim\Q}[\sum_{i\in \agents}{p_i(A)}]$, depending on whether $\p$ is an A-, B-, or C-payment vector, respectively. Subsidy minimization is the generalization of the problem that was recently studied for deterministic allocations and envy-freeness in~\citep*{BDN+20, CI20, HS19}.

Our second optimization problem is called {\em utility maximization} and can be thought of as an extension of rent division~\citep*{GMPZ17} to lotteries and iEF. There is a rent $R$ and the payments are contributions from the agents that, in expectation, should sum up to $R$. So, the entries in the payment vector $\p$ are constrained to be non-positive. The goal is to find an iEF allocation and accompanying payments, such that the minimum expected utility over all agents is maximized. Depending on whether the problem asks for A-, B-, or C-payments, the utility of agent $i\in \agents$ from allocation $A$ is $v_i(A_i)+p_i$, $v_i(A_i)+p(A_i)$, and $v_i(A_i)+p_i(A)$, respectively.

As we will see in Section~\ref{sec:opt}, both subsidy minimization and utility maximization admit much better solutions compared to their versions with deterministic allocations and envy-freeness that had been previously studied in the literature. In addition, the quality of solutions depends on the type of payments. We demonstrate that there is no general advantage of A- or B-payments against each other; this justifies the importance of both types of payments. Clearly, C-payments allow for the best possible solutions as they generalize both A- and B-payments. In Section~\ref{sec:compute-payments}, we restrict our focus on matching instances and show how to solve subsidy minimization and utility maximization efficiently, by exploiting the machinery we developed in Section~\ref{sec:computing}.

\subsection{A characterization for A-payments}\label{sec:char}
We now extend the notion of envy-freeability from recent works focusing on the use of subsidies in fair division settings (e.g., see~\citep*{BDN+20, CI20, HS19}), and earlier studies in rent division (e.g., see~\citep*{ADG91,S09}). Given an allocation instance and a lottery $\Q$ over allocations, we say that $\Q$ is {\em iEF-able} with A-payments if there is a vector $\p$ of A-payments so that the pair ($\Q,\p$) is iEF. Even though we will not need these terms here, we can define the term iEF-ability with B- or C-payments analogously.

Our main result in this section (Theorem~\ref{thm:characterization}) will be a characterization of the lotteries that are iEF-able with A-payments. The notion of the {\em interim envy graph} will be very useful; it extends the notion of the envy-graph that is central in the characterization of envy-freeable allocations (see, e.g.,~\citep*{HS19}).

\begin{defn}\label{defn:envy-graph}
Given a lottery $\Q$ over allocations of the items in set $\items$ to the agents in set $\agents$, the interim envy graph $\EG(\Q,\agents,\items)$ is a complete directed graph with $n$ nodes corresponding to the agents of $\agents$, and edge weights defined as
\begin{align*}
w(i,k) &= \max_{\substack{S \subseteq \items:\\\Pr_{A\sim \Q}[A_i=S]>0}}{\left\{\E_{A\sim \Q}[v_i(A_k)|A_i=S] - v_i(S)\right\}}
\end{align*}
for every directed edge $(i,k)$.
\end{defn}

Our characterization follows; it extends well-known characterizations for deterministic envy-freeable allocations, e.g., see~\citep*{HS19,GMPZ17}. A quick comparison reveals that the second condition in Theorem~\ref{thm:characterization} is much less restrictive than an analogous condition for envy-freeability, which asserts that envy-freeable allocations locally maximize the utilitarian social welfare among all possible redistributions of the bundles to the agents. This justifies our claim that the space of iEF-able lotteries is quite rich. 

\begin{theorem}\label{thm:characterization}
For a lottery $\Q$ over allocations, the following statements are equivalent:
\renewcommand{\labelenumi}{(\roman{enumi})}
\begin{enumerate}
    \item $\Q$ is iEF-able with A-payments.
    \item For every agent $i\in \agents$, let $S_i$ be any bundle of items that is allocated to agent $i$ with positive probability according to $\Q$. Also, let $\sigma:\agents\rightarrow \agents$ be any permutation of agents. Then,
    \begin{align*}
        \sum_{i\in \agents}{v_i(S_i)} & \geq \sum_{i\in \agents}{\E_{A\sim \Q}[v_i(A_{\sigma(i)})|A_i=S_i]}.
    \end{align*}
    \item The interim envy graph $\EG(\Q,\agents,\items)$ has no cycle of positive weight.
\end{enumerate}
\end{theorem}
\begin{proof}
{\bf (i) $\Rightarrow$ (ii).}
First, assume that statement (i) holds, i.e., that lottery $\Q$ is iEF-able with A-payments. This implies that there exists a vector $\p$ of A-payments so that the pair ($\Q,\p$) is iEF. By applying the iEF condition (Definition~\ref{defn:A-payments}) for agent $i\in \agents$, agent $\sigma(i)\in \agents$, and bundle $S_i$, we get:
\begin{align*}
v_i(S_i)+p_i \geq \E_{A\sim \Q}[v_i(A_{\sigma(i)})|A_i=S_i]+p_{\sigma(i)}.
\end{align*}
By rearranging and summing these inequalities over all agents $i\in \agents$, we get
\begin{align}\label{eq:ief-sum}
\sum_{i\in \agents}{\left(v_i(S_i)- \E_{A\sim\Q}[v_i(A_{\sigma(i)})|A_i=S_i]\right)} & \geq \sum_{i\in \agents}{\left(p_{\sigma(i)}-p_i\right)}.
\end{align}
Since $\sigma$ is a permutation over the agents, it is $\sum_{i\in\agents}{p_{\sigma(i)}}=\sum_{i\in\agents}{p_i}$ and the RHS of (\ref{eq:ief-sum}) is equal to $0$. Then, inequality (\ref{eq:ief-sum}) is equivalent to statement (ii).

{\bf (ii) $\Rightarrow$ (iii).}
We now show that if the interim envy graph $\EG(\Q,\agents,\items)$ had a cycle of positive weight, this would violate statement (ii) for a particular permutation, $\sigma^*$, of the agents and a selection of bundles $S^*_1$, $S^*_2$, ..., $S^*_n$ such that bundle $S^*_i$ is given to agent $i$ with positive probability for each agent $i\in \agents$ under $\Q$.
Let $C$ be such a positive-weight cycle, consisting of the agents $i_0$, $i_1$, ..., $i_{|C|-1}$. Define $\sigma^*$ to be the permutation of agents defined as $\sigma^*(i)=i$ for each agent $i$ who does not belong to cycle $C$ and as $\sigma^*(i_j)=i_{j+1 \bmod |C|}$ for each agent $i_j$ (with $j=0, 1, ..., |C|-1$) in the cycle $C$.

Observe that $\E_{A\sim\Q}[v_{A_{\sigma^*(i)}}|A_i=S_i]=v_i(S_i)$ for every agent $i$ not belonging to $C$ and every set $S_i$ allocated to $i$ with positive probability. Hence,
\begin{align}\label{eq:out-of-C}
\sum_{i\in \agents\setminus C}{v_i(S_i)} &= \sum_{i\in \agents \setminus C}{\E_{A\sim \Q}[v_i(A_{\sigma^*(i)})|A_i=S_i]}.
\end{align}
Furthermore, define the bundles $S^*_i$ for each agent $i\in \agents$ as follows. For each agent $i$ that does not belong to $C$, let $S^*_i$ be any bundle that $i$ gets with positive probability according to $\Q$. For agent $i_j$ (with $j=0, 1, ..., |C|-1$) of the cycle $C$, define
\begin{align*}
S^*_{i_j} &= \argmax_{\substack{S \subseteq \items:\\\Pr_{A\sim \Q}[A_{i_j}=S]>0}}{\left\{\E_{A\sim \Q}[v_{i_j}(A_{i_{j+1 \bmod |C|}})|A_{i_j}=S] - v_{i_j}(S)\right\}}.
\end{align*}
Since the total weight of the edges of the interim envy graph $\EG(\Q,\agents,\items)$ is positive, we have (using the definition of edge weight in the interim envy graph in Definition~\ref{defn:envy-graph})
\begin{align*}
0 &< \sum_{j=0}^{|C|-1}{\max_{\substack{S \subseteq \items:\\\Pr_{A\sim \Q}[A_{i_j}=S]>0}}\left\{\E_{A\sim\Q}[v_{i_j}(A_{i_{j+1 \bmod |C|}})|A_{i_j}=S]-v_{i_j}(S)\right\}}\\
&= \sum_{j=0}^{|C|-1}{\left(\E_{A\sim\Q}[v_{i_j}(A_{i_{j+1 \bmod |C|}})|A_{i_j}=S^*_{i_j}]-v_{i_j}(S^*_{i_j})\right)}\\
&= \sum_{i\in C}{\left(\E_{A\sim\Q}[v_i(A_{\sigma^*(i)})|A_i=S^*_{i}]-v_i(S^*_{i})\right)},
\end{align*}
and, equivalently,
\begin{align}\label{eq:in-C}
\sum_{i\in C}{v_i(S^*_{i})} &< \sum_{i\in C}{\E_{A\sim \Q}[v_i(A_{\sigma^*(i)})|A_i=S^*_{i}]}.
\end{align}
Now, by summing equation (\ref{eq:in-C}) with equation (\ref{eq:out-of-C}) for $S_i=S^*_i$, we get a contradiction to statement (ii) for the particular permutation $\sigma^*$ and the given selection of sets $S^*_i$ for $i\in \agents$.

{\bf (iii) $\Rightarrow$ (i).}
Finally, assuming that the interim envy graph $\EG(\Q,\agents,\items)$ has no positive cycle, we will show that $\Q$ is iEF-able with A-payments. For every agent $i\in \agents$, define her payment $p_i$ to be the total weight in the longest (in terms of total length) path from node $i$ in the interim envy graph (this is well defined because $\EG(\Q,\agents,\items)$ has no positive cycles). Then, for every pair of agents $i,k\in \agents$, and for every bundle $S'$ agent $i$ gets with positive probability under $\Q$, we have (using the definition of the edge weights in the interim envy graph)
\begin{align*}
p_i &\geq p_k+w(i,k)\\
&= p_k+\max_{\substack{S \subseteq \items:\\\Pr_{A\sim \Q}[A_i=S]>0}}{\left\{\E_{A\sim \Q}[v_i(A_k)|A_i=S] - v_i(S)\right\}}\\
&\geq p_k+\E_{A\sim \Q}[v_i(A_k)|A_i=S'] - v_i(S').
\end{align*}
By rearranging we get that the iEF condition for agent $i$ with respect to agent $k$ and the bundle $S'$ is satisfied.
\end{proof}
\subsection{Contrasting A-payments with B-payments}\label{sec:opt}
We now attempt a comparison between different types of payments. First, we remark that iEF lotteries with A-payments (or B-payments) can yield much lower total expected subsidies and much higher minimum expected utility for utility maximization instances, compared to envy-free allocations with payments. This should be clear given Lemma~\ref{lem:ief-vs-ef}; we give explicit bounds on subsidy minimization and utility maximization with the next example.

\begin{example}\label{ex:ief-vs-ef}
Consider the instance at the left of Table~\ref{tab:ef-vs-ief-payments}. By the characterization of~\cite{HS19}, we know that in matching instances only the most efficient allocation of items to agents is envy-freeable. Therefore, $a$-$b$-$c$ and $a$-$c$-$b$ are the only envy-freeable allocations and this is possible with a payment of $1/3$ to agent 1 (or, to item $a$) and no payments to the other two agents (or, to the other two items). In contrast, the lottery that has both allocations in its support with equal probability is iEF without any payments.

\begin{table}[h]
\centering
\caption{An instance of subsidy minimization (left) and utility maximization (right) with three agents where envy-freeness is inferior to iEF with A- and B-payments. \label{tab:ef-vs-ief-payments}}
{\begin{tabular}{l r}
{\begin{tabular}{c| cccc}
\hline
&$a$&$ b$&$c$\\
\hline
$1$&$1/3$&$2/3$&$0$\\
$2$&$0$&$1/2$&$1/2$\\
$3$&$0$&$1/2$&$1/2$\\
\hline
\end{tabular}}
&\quad\quad
{\begin{tabular}{c| cccc}
\hline
&$a$&$ b$&$c$\\
\hline
$1$&$1/4$&$3/4$&$0$\\
$2$&$0$&$1/2$&$1/2$\\
$3$&$0$&$1/2$&$1/2$\\
\hline
\end{tabular}}
\end{tabular}}
{}
\end{table}

Now, consider the matching instance at the right of Table~\ref{tab:ef-vs-ief-payments} and let $R=1$. The only envy-freeable allocations are $a$-$b$-$c$, $a$-$c$-$b$, $b$-$a$-$c$ and $b$-$c$-$a$. The rent shares that make the first two allocations EF are $0$, $1/2$, and $1/2$ to agents 1, 2, and 3, respectively, while they are $1/2$, $0$, and $1/2$ and $1/2$, $1/2$, and $0$ to agents 1, 2, and 3, respectively, for the last two allocations. Note that agents 2 and 3 obtain utility $0$ under these payments. In contrast, the lottery that has both allocations in its support with equal probability is iEF with payments $1/6$, $5/12$, and $5/12$ by agents 1, 2, and 3, (or, to items $a$, $b$, $c$) which sum up to $1$. The (expected) utility of each agent is then $1/12$.
\end{example}

We now compare A-payments to B-payments in terms of the quality of the solutions they yield the two optimization problems. In particular, the proof of Theorem~\ref{thm:a-vs-b} presents instances where B-payments are superior to A-payments. The proof of Theorem~\ref{thm:a-vs-b} shows two instances, of subsidy minimization and utility maximization, respectively. For subsidy minimization, the expected amount of subsidies achieved with B-payments is arbitrarily close to $0$, while A-payments need a constant amount of subsidies. Similarly, in the utility maximization instance, the minimum expected agent utility is arbitrarily close to $0$ with A-payments and considerably higher with B-payments. Perhaps surprisingly, there are also instances where A-payments are preferable to B-payments, as shown in the proof of Theorem~\ref{thm:b-vs-a}.

\begin{theorem}\label{thm:a-vs-b}
The solution of subsidy minimization and utility maximization with B-payments can be strictly better than the solution of the corresponding problems with A-payments.
\end{theorem}

\begin{proof}
For subsidy minimization, consider the matching instance at the left of Table~\ref{tab:ief-A-vs-B-payments}.
Under B-payments, consider the lottery $\Q$ that selects among allocations $a$-$b$-$c$ and $a$-$c$-$b$ equiprobably and the corresponding payment vector $\p$ such that $p(a) =  \epsilon$, $p(b) = 0$, and $p(c) = 2\epsilon$. It can be verified that the pair ($\Q,\p$) is iEF with a total payment amount of $3\epsilon$. 
Any iEF pair with A-payments, however, would require a total payment of at least $1/3+2\epsilon$. 
To see why, consider an iEF pair ($\Q',\p'$) with a smaller total payment amount. If agent $1$ gets item $c$ with positive probability, then $\Q'$ necessarily includes an allocation where two agents get their unique least valuable item. This leads to a cycle with positive weight in the interim envy graph and, by Theorem \ref{thm:characterization}, $\Q'$ cannot be part of any iEF pair. Therefore, agent $1$ gets either item $a$ or $b$ in $\Q'$ . Note that when $2$ gets $a$ it is always the case that $1$ gets $b$ and, so that agent $2$ does not envy agent $1$, it must be $p_2 \geq 1/2+ \epsilon + p_1>1/3+2\epsilon$; when agent $3$ gets $a$ the analysis is symmetric. So, agent $1$ always gets item $a$ in $\Q'$. If $\Q'$ randomizes over allocations $a$-$b$-$c$ and $a$-$c$-$b$, then we must have both $1/2-\epsilon +p_2\geq 1/2+\epsilon+p_3$ and $1/2-\epsilon +p_3\geq 1/2+\epsilon+p_2$ which cannot occur. 
So, it remains to examine the case of the deterministic allocation $a$-$b$-$c$ (or $a$-$c$-$b$ which is symmetric). Indeed, the allocation $a$-$b$-$c$ with agent-payment vector $(1/3, 0, 2\epsilon)$ is an iEF pair with a total payment of $1/3+2\epsilon$, as desired.

\begin{table}[h]
\centering
\caption{An instance of subsidy minimization (left) and utility maximization (right) with three agents where B-payments are superior, for an arbitrarily small $\epsilon>0$. \label{tab:ief-A-vs-B-payments}}
{\begin{tabular}{l r}
{\begin{tabular}{c| cccc}
\hline
&$a$&$ b$&$c$\\
\hline
$1$&$1/3$&$2/3$&$0$\\
$2$&$0$&$1/2+\epsilon$&$1/2-\epsilon$\\
$3$&$0$&$1/2+\epsilon$&$1/2-\epsilon$\\
\hline
\end{tabular}}
&\quad\quad
{\begin{tabular}{c| cccc}
\hline
&$a$&$ b$&$c$\\
\hline
$1$&$1/4$&$3/4$&$0$\\
$2$&$0$&$1/2+\epsilon$&$1/2-\epsilon$\\
$3$&$0$&$1/2+\epsilon$&$1/2-\epsilon$\\
\hline
\end{tabular}}
\end{tabular}}
{}
\end{table}

For utility maximization and $R=1$, consider the instance at the right of Table~\ref{tab:ief-A-vs-B-payments}. 
Under B-payments, consider the lottery $\Q$ that selects allocations $a$-$b$-$c$ and $a$-$c$-$b$ equiprobably and the payment vector $\p$ such that $p(a) = -1/6$, $p(b) = -5/12-\epsilon$, and $p(c)=-5/12+\epsilon$. One can again verify that the pair ($\Q,\p$) is iEF where the expected utility of each agent is $1/12$.
Any iEF pair with A-payments, however, would lead to a minimum expected utility of at most $\epsilon/3$.  
Indeed, let $\Q'$ and $\p'$ be an iEF pair with minimum expected utility larger than $\epsilon/3$. Again, Theorem \ref{thm:characterization} suggests that it cannot be that, in $\Q'$, at least two agents get an item that they value at $0$ as this would create a cycle of positive weight in the interim envy graph. So, agent $1$ can only be assigned items $a$ or $b$. 
Let agent $2$ get item $a$ with positive probability (the argument for agent $3$ is symmetric). Then agent $1$ must get item $b$ and, hence, the allocation $b$-$a$-$c$ is in $\Q'$. It must be $p_2\geq 1/2+\epsilon+p_1$ and $p_2\geq 1/2-\epsilon+p_3$, so that agent $2$ does not envy $1$ or $3$. By summing the two inequalities and adding $p_2$ to both sides, and since $p_1+p_2+p_3=-1$, we get $p_2\geq 0$. Since payments are non-positive, we obtain $p_2=0$ and, hence, $p_1 = -1/2-\epsilon$ and $p_3 = -1/2+\epsilon$ and the minimum expected utility is at most $0$. So, agent $1$ is always assigned item $a$ in $\Q'$. If $\Q'$ randomizes over allocations $a$-$b$-$c$ and $a$-$c$-$b$, it must be $1/2-\epsilon+p_3\geq 1/2+\epsilon+p_2$ and $1/2-\epsilon+p_2\geq 1/2+\epsilon+p_3$ which cannot occur. So, there remains the deterministic allocation $a$-$b$-$c$ (or $a$-$c$-$b$ which is symmetric). Indeed, the allocation $a$-$b$-$c$ forms an iEF pair with the agent-payment vector $(-2\epsilon/3, -1/2-2\epsilon/3, -1/2+4\epsilon/3)$ with a minimum expected utility of $\epsilon/3$, as desired. 
\end{proof}

\begin{theorem}\label{thm:b-vs-a}
The solution of subsidy minimization and utility maximization with A-payments can be strictly better than the solution of the corresponding problems with B-payments.
\end{theorem}
\begin{proof}
For subsidy minimization, consider the instance at the left of Table~\ref{tab:ief-B-vs-A-payments}.
For A-payments, consider the lottery $\Q$ that selects among allocations $a$-$b$-$c$, $a$-$c$-$b$, and $b$-$c$-$a$ equiprobably, together with the payment vector $\p= (0, 0, 3\epsilon)$. It can be verified that the pair ($\Q,\p$) is iEF with a total payment of $3\epsilon$.
An iEF pair with B-payments, however, would require payments at least $6\epsilon$. Indeed, consider an iEF pair ($\Q',\p'$) with a smaller payment amount. If in $\Q'$ an agent gets with positive probability an item she values at $0$, the total payment is at least $2/5$. Hence $\Q'$ randomizes over (at most) $a$-$b$-$c$, $a$-$c$-$b$, and $b$-$c$-$a$. If $a$-$b$-$c$ is not in $\Q'$, then agent $2$ deterministically gets item $c$. Then, $\Q'$ cannot include both $a$-$c$-$b$ and $b$-$c$-$a$ as, then, both $2/5+p_1\geq 3/5+p_3$ and $1/3-\epsilon +p_3 \geq 1/3+2\epsilon+p_1$ must hold so that agents $1$ and $3$ do not envy each other. Regardless of whether $\Q'$ contains just $a$-$c$-$b$ or just $b$-$c$-$a$, the payment is at least $1/5$ so that agent $2$ is not envious. Similarly, if $b$-$c$-$a$ is not in $\Q'$, then agent $1$ deterministically gets item $a$ and the analysis is symmetric, leading to a payment of at least $1/5$ so that agent $1$ is not envious. 
If $a$-$c$-$b$ is not in $\Q'$, then a payment of at least $1/5$ is again needed. Finally, when $\Q'$ consists of all three allocations, this yields $1/3-\epsilon+p(a)\geq 1/3+2\epsilon+p(b)$ and $1/3-\epsilon+p(c)\geq 1/3+2\epsilon+p(b)$ so that agent $3$ is not envious. This gives a lower bound of $6\epsilon$ on the payment.
Indeed, the lottery that selects equiprobably among allocations $a$-$b$-$c$, $a$-$c$-$b$, and $b$-$c$-$a$ together with payments $p(a) = 3\epsilon$, $p(b)=0$, and $p(c) = 3\epsilon$ is an iEF pair with a total payment of $6\epsilon$. 

\begin{table}[h]
\centering
\caption{An instance of subsidy minimization (left) and utility maximization (right) with three agents where A-payments are preferable, for an arbitrarily small $\epsilon>0$. \label{tab:ief-B-vs-A-payments}}
{\begin{tabular}{l r}
{\begin{tabular}{c| cccc}
\hline
&$a$&$ b$&$c$\\
\hline
$1$&$2/5$&$3/5$&$0$\\
$2$&$0$&$3/5$&$2/5$\\
$3$&$1/3-\epsilon$&$1/3+2\epsilon$&$1/3-\epsilon$\\
\hline
\end{tabular}}
&\quad\quad
{\begin{tabular}{c| cccc}
\hline
&$a$&$ b$&$c$\\
\hline
$1$&$9/20-\epsilon$&$11/20+\epsilon$&$0$\\
$2$&$0$&$11/20+\epsilon$&$9/20-\epsilon$\\
$3$&$3/10$&$2/5$&$3/10$\\
\hline
\end{tabular}}
\end{tabular}}
{}
\end{table}

For utility maximization and $R=1$, consider the instance at the right of Table~\ref{tab:ief-B-vs-A-payments}.
For A-payments, consider the lottery $\Q$ that selects allocation $a$-$b$-$c$ with probability $3/10$, $a$-$c$-$b$ with probability $2/5$, and $b$-$c$-$a$ with probability $3/10$, together with the payment vector $\p=(-11/30-\epsilon/2, -11/30-\epsilon/2, -4/15+\epsilon)$. It can be verified that the pair ($\Q,\p$) is iEF with minimum expected utility of $11/150+\epsilon$.
An iEF pair with B-payments leads to a minimum expected utility of at most $2\epsilon/3$. 
Assume otherwise and consider a pair ($\Q',\p'$) with minimum expected utility greater than $2\epsilon/3$. $\Q'$ cannot include $c$-$a$-$b$ as then it should hold $p(a)>p(c)$ and $p(c)>p(a)$. Furthermore, $\Q'$ cannot include just the allocation $a$-$c$-$b$ as then it must hold $p'(a)-p'(b)\geq 1/10+2\epsilon$ and $p'(a) -p'(b) \leq 1/10$ so that agents $1$ and $3$ do not envy each other.
Furthermore, if $\Q'$ includes $a$-$c$-$b$, then, since $c$-$a$-$b$ is not in $\Q'$, it must be $p'(a) = p'(c) =-3/10$ and $p'(b) = -2/5$ so that agent $3$ is not envious; the minimum expected utility is at most $0$ in this case. So, $\Q'$ randomizes over allocations $a$-$b$-$c$ and $b$-$c$-$a$; any such choice yields $p(a) = -2/5-4\epsilon/3$, $p'(b) = p'(c) = -3/10+2\epsilon/3$ and the minimum expected utility is $2\epsilon/3$.
\end{proof}

\subsection{Computing $\epsilon$-iEF allocations with C-payments}\label{sec:compute-payments}
In this section, we show how to solve efficiently subsidy minimization and utility maximization when we are allowed to use C-payments (and sharp approximations of iEF). Our algorithms use linear programming and the machinery we developed in Section~\ref{sec:computing}. Our result for subsidy minimization is the following.

\begin{theorem}\label{thm:sm}
Let $\epsilon>0$. Consider an instance of subsidy minimization with C-payments and let $\opt$ be the value of its optimal solution. Our algorithm computes a lottery $\Q$ and a C-payment vector $\p$ of expected value $\opt$, so that the pair ($\Q,\p$) is $\epsilon$-iEF with C-payments.
\end{theorem}

We begin by defining a linear program for computing an iEF pair of lottery and C-payment vector. We use the variable vector $\x$ to denote the lottery. By Definition~\ref{defn:C-payments}, the iEF constraint for agent $i\in \agents$ who is assigned item $j\in \items$ with positive probability under lottery $\x$ and another agent $k\in\agents\setminus\{i\}$ is
\begin{align}\label{eq:ief-condition-p}
    v_i(j)+\E_{b\sim \x}[p_i(b)|b(i)=j] &\geq \E_{b\sim \x}[v_i(b(k))+p_k(b)|b(i)=j].
\end{align}
By multiplying the left-hand-side of (\ref{eq:ief-condition-p}) with the positive probability $\prm_{b\sim\x}[b(i)=j]$, we get
\begin{align}\nonumber
    &v_i(j)\cdot \prm_{b\sim\x}[b(i)=j]+\E_{b\sim \x}[p_i(b)|b(i)=j]\cdot \prm_{b\sim\x}[b(i)=j]\\\label{eq:ief-condition-p-LHS}
     &= \sum_{b\in \M_{ij}}{x(b)\cdot v_i(j)}+\sum_{b\in \M_{ij}}{x(b)\cdot p_i(b)}.
\end{align}
By multiplying the right-hand-side of (\ref{eq:ief-condition-p}) again with $\prm_{b\sim\x}[b(i)=j]$, we obtain
\begin{align}\nonumber
    &\E_{b\sim \x}\left[v_i(b(k))+p_k(b)|b(i)=j]\cdot \prm_{b\sim\x}[b(i)=j\right]\\\label{eq:ief-condition-p-RHS}
    &= \sum_{b\in \M_{ij}}{x(b)\cdot v_i(b(k))}+\sum_{b\in \M_{ij}}{x(b)\cdot p_k(b)}.
\end{align}
Hence, using (\ref{eq:ief-condition-p-LHS}) and (\ref{eq:ief-condition-p-RHS}), (\ref{eq:ief-condition-p}) yields
\begin{align}\label{eq:ief-condition-p-final}
\sum_{b\in \M_{ij}}{\left(x(b)\cdot \left(v_i(j)-v_i(b(k))\right)+x(b)\cdot p_i(b)-x(b)\cdot p_k(b)\right)} &\geq 0.
\end{align}
Notice the products $x(b)\cdot p_i(b)$ and $x(b)\cdot p_k(b)$ in the above expression. In such terms, both factors are unknowns that we have to compute. As $p_i(b)$ always appears multiplied with $x(b)$ in the above expressions, we can avoid non-linearities by introducing the variable $t_i(b)$ for every agent $i\in \agents$ and matching $b\in \M$ to be thought of as equal to $x(b)\cdot p_i(b)$. With this interpretation in mind, equation (\ref{eq:ief-condition-p-final}) becomes
\begin{align*}
\sum_{b\in \M_{ij}}{\left(x(b)\cdot (v_i(j)-v_i(b(k)))+t_i(b)-t_k(b)\right)} \geq 0.
\end{align*}
Furthermore, our objective is to minimize $\sum_{b\in M}\sum_{i\in\agents}{t_i(b)}$ since \begin{align*}
\sum_{b\in \M}{\sum_{i\in \agents}{t_i(b)}}=\sum_{b\in \M}{\sum_{i\in \agents}{x(b)\cdot p_i(b)}}=\E_{b\sim \x}\left[\sum_{i\in \agents}{p_i(b)}\right].
\end{align*}
Summarizing, our linear program for subsidy minimization is
\begin{equation}\label{LP-min-payments}
\begin{array}{ll@{}ll}
\text{minimize}  & \displaystyle\sum\limits_{b\in \M}\sum\limits_{i\in \agents} t_i(b) & \\
\text{subject to}& \displaystyle\sum\limits_{b\in \M_{ij}}   \left(x(b)\cdot (v_i(j)-v_i(b(k)))+t_i(b)-t_k(b)\right) \geq 0, & i\in \agents, j\in \items, k\in \agents\setminus\{i\}\\
 & \displaystyle\sum\limits_{b\in \M}   x(b) =1 &\\
 &                                                \quad\quad x(b) \geq 0, & b\in \M\\
 &                                                \quad\quad t_i(b) \geq 0, & b\in \M, i\in \agents
\end{array}
\end{equation}

The proof of our next lemma (Lemma~\ref{lem:LP-min-payments}) follows along similar lines to the approach we followed in Section~\ref{sec:computing}. A solution of the linear program~(\ref{LP-min-payments}) naturally gives an iEF pair of lottery $\x$ and C-payments vector $\p$ when $x(b)=0$ for a matching $b\in\M$ implies that $t_i(b)=0$ for every agent $i\in\agents$. Unfortunately, we have not excluded the case that $x(b)=0$ and $t_i(b)>0$ in the solution of the linear program (\ref{LP-min-payments}). We take care of such cases by modifying the solution returned by our algorithm and violating the iEF condition marginally. 

\begin{lemma}\label{lem:LP-min-payments}
The linear program (\ref{LP-min-payments}) can be solved in polynomial time.
\end{lemma}
\begin{proof}
We follow our approach from Section~\ref{sec:computing} and solve the linear program (\ref{LP-min-payments}) by applying the ellipsoid method to its dual linear program:
\begin{equation}\label{dual-LP-min-payments}
\begin{array}{ll@{}ll}
\text{maximize}  & z & \\
\text{subject to}& z + \displaystyle\sum\limits_{\substack{(i,j)\in b\\k\in \agents \setminus\{i\}}} {(v_i(j)-v_i(b(k)))\cdot y(i,j,k)} \leq 0, & \,\,b \in\M\\
 & \sum\limits_{k\in \agents\setminus\{i\}}{\left(y(i,b(i),k)-y(k,b(k),i)\right)} -1 \leq 0, & \,\,i\in \agents, b\in \M\\
 & y(i,j,k) \geq 0, i\in \agents, j\in \items, k\in\agents\setminus\{i\}\\
\end{array}
\end{equation}
The first set of constraints of the dual linear program (\ref{dual-LP-min-payments}) is very similar to the constraints of the linear program (\ref{dual-LP1}) for the problem without payments. The separation oracle will check whether a given assignment of values to the dual variables satisfies the first set of constraints by solving an appropriately defined instance of 2EBM. The second set of constraints is simpler; the separation oracle will perform several classical bipartite matching computations.

In particular, consider the instance of 2EBM consisting of the complete bipartite graph $G(\agents, \items,\agents\times\items)$ with edge-pair weight
\begin{align*}
\psi(i,j,k,\ell) &= (v_i(j)-v_i(\ell))\cdot y(i,j,k)+\frac{z}{n(n-1)}
\end{align*}
for each quadruple $(i,j,k,\ell)\in \X$. Then, the total weight $\Psi(b)$ over all edge-pairs of matching $b\in \M$ is equal to the LHS of the constraint corresponding to matching $b$ among those in the first set of constraints. Furthermore, for agent $i\in \agents$, consider the complete bipartite graph $G^i(\agents,\items,\agents\times\items)$ with weight $c(i,\ell)=\sum_{k\in \agents\setminus\{i\}}{y(i,\ell,k)}-\frac{1}{n}$ for every $\ell\in \items$ and $c(k,\ell)=-y(k,\ell,i)-\frac{1}{n}$, for every $k\in \agents\setminus\{i\}$ and $\ell\in \items$. Then, the total weight of edges in a perfect matching $b\in \M$ is equal to the LHS of the second constraint of the dual linear program (\ref{dual-LP-min-payments}) associated with matching $b\in \M$ and agent $i\in \agents$.

Hence, our separation oracle for the dual linear program (\ref{dual-LP-min-payments}) first solves the instance of 2EBM and computes a maximum edge-pair-weighted perfect matching $b^*$ in $G$. Then, it computes a maximum edge-weighted perfect matching $b^*_i$ in graph $G^i$ for $i\in \agents$. If $\Psi(b^*)>0$, then it reports the constraint associated with matching $b^*\in \M$ from the first set of constraints of the dual linear program (\ref{dual-LP-min-payments}) as violating. If the total edge-weight of some matching $b^*_i$ is strictly positive, then it reports the constraint associated with agent $i\in \agents$ and matching $b^*\in \M$ from the second set of constraints of (\ref{dual-LP-min-payments}) as violating. Otherwise, i.e., if all perfect matchings have non-positive total edge-pair or edge weight, the separation oracle reports that no constraint is violated.
\end{proof}

The next lemma completes the proof of Theorem~\ref{thm:sm}.

\begin{lemma}\label{lem:violate-ief}
For every $\epsilon>0$, given any extreme solution to the linear program (\ref{LP-min-payments}) of objective value $\opt$, an $\epsilon$-iEF lottery $\x'$ with C-payments $\p'$ of total expected value $\opt$ can be computed in polynomial time.
\end{lemma}
\begin{proof}
Let us assume that we have an extreme solution to linear program (\ref{LP-min-payments}). Since the number of non-trivial constraints is polynomial (at most $n^3$) and the rest just require that the variables $x(b)$ and $t_i(b)$ are non-negative, we conclude that the number of variables that have non-zero values are at most polynomially many (i.e., at most $n^3$).

Let $V_{\max}=\max_{i\in \agents, j\in\items}{v_i(j)}$ and set $\delta=\frac{\epsilon}{2V_{\max}}$. Denote by $\K_1$ the set of matchings $b\in \M$ with $x(b)>0$ and by $\K_2$ the set of matchings $b$ such that $x(b)=0$ and $t_i(b)>0$ for some agent $i\in \agents$.
Our transformation sets $x'(b)=(1-\delta)x(b)$ for every matching $b\in \K_1$ and $x'(b)=\delta/|\K_2|$ for every matching $b\in \K_2$. The remaining variables are left intact.

We first prove that $\x'$ is a lottery. Indeed, we have
\begin{align*}
\sum_{b\in \M}{x'(b)} &= \sum_{b\in \K_1}{x'(b)}+\sum_{b\in \K_2}{x'(b)}=(1-\delta)\sum_{b\in \K_1}{x(b)}+\sum_{b\in \K_2}{\frac{\delta}{|\K_2|}}=1.
\end{align*}
To see this, note that $\sum_{b\in \K_1}{x(b)}=1$. Furthermore, observe that
\begin{align*}
&\sum_{b\in\M_{ij}}{x'(b)\cdot (v_i(j)-v_i(b(k)))}\\
&= (1-\delta)\sum_{b\in\M_{ij}\cap \K_1}{x(b)\cdot (v_i(j)-v_i(b(k)))}+\sum_{b\in \M_{ij}\cap \K_2}{\frac{\delta}{|\K_2|}\cdot (v_i(j)-v_i(b(k)))}\\
&\geq \sum_{b\in\M_{ij}}{x(b)\cdot (v_i(j)-v_i(b(k)))}-\delta\sum_{b\in \M_{ij}\cap \K_1}{x(b)\cdot V_{\max}}-\delta\cdot V_{\max}\\
&\geq \sum_{b\in\M_{ij}}{x(b)\cdot (v_i(j)-v_i(b(k)))}-2\delta\cdot V_{\max}\\
&= \sum_{b\in\M_{ij}}{x(b)\cdot (v_i(j)-v_i(b(k)))}-\epsilon.
\end{align*}
Hence, evaluating the first constraint of linear program (\ref{LP-min-payments}) for lottery $x'$, we have (using the fact that the lottery $\x$ satisfies the constraint)
\begin{align*}
&\sum_{b\in\M_{ij}}{\left(x'(b)\cdot (v_i(j)-v_i(b(k)))+t_i(b)-t_k(b)\right)}\\
&\geq \sum_{b\in\M_{ij}}{\left(x(b)\cdot (v_i(j)-v_i(b(k)))+t_i(b)-t_k(b)\right)}-\epsilon\\
&\geq -\epsilon.
\end{align*}
To construct the payment vector $\p$ that makes $\x'$ $\epsilon$-iEF, it suffices to set $p_i(b)=t_i(b)/x'(b)$ for every $b\in \K_1\cup \K_2$ and $p_i(b)=0$ otherwise.
\end{proof}

Our result for utility maximization is the following.

\begin{theorem}\label{thm:umax}
Let $\epsilon>0$. Consider an instance of utility maximization with C-payments and let $\opt$ be the value of its optimal solution. Our algorithm computes a lottery $\Q$ and a C-payment vector $\p$ of expected value at least $\opt-\epsilon$, so that the pair ($\Q,\p$) is $\epsilon$-iEF with C-payments.
\end{theorem}

Again, our algorithm uses the solution of an appropriate linear program defined as follows:
\begin{equation}\label{LP-max-utility}
\begin{array}{ll@{}ll}
\text{maximize}  & q & \\
\text{subject to}& q - \displaystyle\sum\limits_{b\in \M}  x(b)\cdot v_i(b(i))+\displaystyle\sum\limits_{b\in \M}{t_i(b)} \leq 0& i\in \agents\\
& \displaystyle\sum\limits_{b\in \M_{ij}}  \left(x(b)\cdot (v_i(j)-v_i(b(k)))-t_i(b)+t_k(b)\right) \geq 0, &   i\in \agents,  j\in \items, k\in \agents\setminus\{i\}\\
 & \displaystyle\sum\limits_{b\in \M}   x(b) =1 &\\
 &  \displaystyle\sum\limits_{b\in \M}{\sum\limits_{i\in \agents}{t_i(b)}} = R & \\
 &                                                \quad\quad x(b) \geq 0, & b\in \M\\
 &                                                \quad\quad t_i(b) \geq 0, & b\in \M, i\in \agents
\end{array}
\end{equation}
Compared to our previous linear program (\ref{LP-min-payments}), this one has the extra variable $q$, which denotes the minimum expected utility over all agents. Note that, here, we require $t_i(b)$ to be non-negative and interpret $t_i(b)$ as $-x(b)\cdot p_i(b)$; recall that payments in utility maximization are non-positive. The objective is to maximize $q$ and the first set of constraints guarantees that $q$ is at least the expected utility $\E_{b\sim\x}[v_i(b(i)+p_i(b)]$ for every agent $i\in\agents$. Indeed, observe that
\begin{align*}
\E_{b\sim \x}[v_i(b(i))+p_i(b)] &=\sum_{b\in \M}{x(b)\cdot (v_i(b(i))+p_i(b))}=\sum_{b\in \M}{x(b)\cdot v_i(b(i))}-\sum_{b\in \M}{t_i(b)}.
\end{align*}
The second set of constraints captures the iEF requirements. Due to the definition of the payment variables, the signs of $t_i(b)$ and $t_k(b)$ in the second constraint are different compared to the linear program (\ref{LP-min-payments}). The third set of constraints restricts $\x$ to be a lottery and the fourth one guarantees that the payments contribute to the rent.

The proof of Theorem~\ref{thm:umax} follows by the next two lemmas. Lemma~\ref{lem:LP-max-utility} uses again the approach we developed in Section~\ref{sec:computing} to compute a solution to the linear program (\ref{LP-max-utility}), and Lemma~\ref{lem:violate-ief-umax} modifies the solution of the linear program to come up with a pair of lottery and corresponding C-payments that marginally violate iEF and achieve a marginally lower objective value.

\begin{lemma}\label{lem:LP-max-utility}
The linear program (\ref{LP-max-utility}) can be solved in polynomial time.
\end{lemma}

\begin{proof}
Once again, we solve the linear program (\ref{LP-max-utility}) by applying the ellipsoid method to its dual linear program, which is now the following:
\begin{equation}\label{dual-LP-max-utility}
\begin{array}{ll@{}ll}
\text{maximize}  & z +Rg & \\
\text{subject to}& z + \displaystyle\sum\limits_{i\in \agents}{w_i\cdot v_i(b(i))}+\displaystyle\sum\limits_{\substack{(i,j)\in b\\k\in \agents \setminus\{i\}}} {(v_i(j)-v_i(b(k)))\cdot y(i,j,k)} \leq 0, & \,\, b \in\M\\
 & \displaystyle\sum\limits_{k\in \agents\setminus\{i\}}{\left(y(k,b(k),i)-y(i,b(i),k)\right)} -w_i+g \leq 0, & \,\, i\in \agents, b\in \M\\
 & \displaystyle\sum\limits_{i\in \agents}{w_i} \geq 1 &\\
 & y(i,j,k) \geq 0, i\in \agents, j\in \items, k\in\agents\setminus\{i\}\\
\end{array}
\end{equation}
Our separation oracle will check whether a given assignment of values to the dual variables satisfies the first set of constraints by solving an instance of 2EBM consisting of a complete bipartite graph $G(\agents,\items,\agents\times\items)$ with edge-pair weight
\begin{align*}
\psi(i,j,k,\ell) &= (v_i(j)-v_i(\ell))\cdot y(i,j,k) + \frac{w_i \cdot v_i(j)}{n-1}+\frac{z}{n(n-1)}
\end{align*}
for each quadruple $(i,j,k,\ell)\in\X$. Indeed, the total edge-pair weight $\Psi$ over all edge-pairs of matching $b\in \M$ is equal to the LHS of the corresponding first constraint of the dual linear program (\ref{dual-LP-max-utility}). 

For the second set of constraints, the separation oracle will solve classical maximum bipartite matching problems. In particular, for agent $i\in\agents$, it will compute a perfect matching of maximum total edge weight for the instance consisting of the complete bipartite graph $G^i(\agents,\items,\agents\times\items)$ with edge weight 
\begin{align*}
c(i,\ell) &=-\sum_{k\in \agents\setminus\{i\}}{y(i,\ell,k)}-\frac{w_i}{n}+\frac{g}{n}
\end{align*}
for every $\ell\in \items$ and 
\begin{align*}
c(k,\ell) &=y(k,\ell,i)-\frac{w_i}{n}+\frac{g}{n},
\end{align*}
for every $k\in \agents\setminus\{i\}$ and $\ell\in \items$. Again, the total weight of edges in a perfect matching $b\in \M$ is equal to the LHS of the second constraint of the dual linear program (\ref{dual-LP-max-utility}) associated with matching $b\in \M$ and agent $i\in \agents$.

Finally, checking whether the dual variables satisfy the third constraint is trivial.
\end{proof}

\begin{lemma}\label{lem:violate-ief-umax}
For every $\epsilon>0$, given an extreme solution to the linear program (\ref{LP-max-utility}) of objective value $\opt$, an $\epsilon$-iEF lottery $\x'$ with C-payments $\p'$ of total expected value at least $\opt-\epsilon$ can be computed in polynomial time.
\end{lemma}

\begin{proof}
We repeat the proof of Lemma~\ref{lem:violate-ief} to compute the lottery $\x'$ and the payment vector $\p$ so that the pair $\x',\p$ is $\epsilon$-iEF. Set $\delta=\frac{\epsilon}{V_{\max}}$. The expected value of agent $i\in \agents$ for her item is
\begin{align*}
\E_{b\sim\x'}[v_i(b(i))] &= \sum_{b\in\M}{x'(b)\cdot v_i(b(i))}\geq (1-\delta) \sum_{b\in \M}{x(b)\cdot v_i(b(i))}\\
&\geq \sum_{b\in \M}{x(b)\cdot v_i(b(i))}-\delta\sum_{b\in\M}{x(b)\cdot V_{\max}}\geq \sum_{b\in \M}{x(b)\cdot v_i(b(i))} - \epsilon.
\end{align*}
The first inequality is due to the definition of $\x'$ (recall that $x'(b)=(1-\delta)x(b)$ for $b\in \K_1$ while $x'(b)\geq 0$ and $x(b)=0$ for $b\in \M\setminus \K_1$), the second one follows by the definition of $V_{\max}$, and the third one follows by the definition of $\delta$ and since $\sum_{b\in\M}{x(b)}=1$. Hence, by the definition of variable $q$ in the solution of the linear program (\ref{LP-max-utility}), we have that the minimum expected utility is
\begin{align*}
\min_{i\in\agents}{\left\{\E_{b\sim \x'}[v_i(b(i)+p_i(b)]\right\}}\geq \min_{i\in\agents}{\left\{\sum_{b\in \M}{x(b)\cdot v_i(b(i))}-\sum_{b\in\M}{t_i(b)}\right\}} - \epsilon\geq q-\epsilon=OPT-\epsilon,
\end{align*}
as desired. 
\end{proof}

\section{Open problems}
\label{sec:open}
We believe that interim envy-freeness can be a very influential fairness notion and can play for lotteries of allocations the role that envy-freeness has played for deterministic allocations. Our work aims to reinvigorate the study of this notion; we hope this will further intensify the study of fairness in random allocations overall. At the conceptual level, iEF can inspire new fairness notions, analogous to known relaxations of envy-freeness, such as envy-freeness up to one (EF1; see~\cite{B11}) and up to any item (EFX; see~\cite{CKMPS19}), that have become very popular recently. Even though it is very tempting, we refrain from proposing additional definitions here.

An appealing feature of iEF lotteries is that they are efficiently computable in matching allocation instances. Of course, besides the importance of the ellipsoid algorithm in theory (see, e.g.,~\cite{GLS88,S86}), our methods have apparent limitations. Combinatorial algorithms for solving the computational problems addressed in Sections~\ref{sec:computing} and~\ref{sec:compute-payments} are undoubtedly desirable. An intermediate first step would be to design a combinatorial algorithm for the maximum edge-pair-weighted perfect bipartite matching problem (2EBM). This could be useful elsewhere, as 2EBM is a very natural problem with possible applications in other contexts.

At the technical level, there is room for several improvements of our results. Our bounds on the price of iEF with respect to the average Nash social welfare have a gap between $\Omega(\sqrt{n})$ (Theorem~\ref{thm:poief-nash-lower}) and $\bigO(n)$ (Theorem~\ref{thm:poief-upper}). Also, in Section~\ref{sec:computing}, we show how to compute iEF lotteries that maximize the expected log-Nash social welfare. Although maximizing log-Nash and average Nash social welfare are equivalent goals for deterministic allocations, this is not the case for lotteries. Computing iEF lotteries of maximum expected average Nash welfare is elusive at this point.

We left for the end the many open problems that are related to iEF with payments. Our characterization in Section~\ref{sec:char} has been used only in the proof of Theorem~\ref{thm:a-vs-b}. It would be interesting to see whether it has wider applicability and, in particular, whether it can lead to efficient algorithms for computing iEF pairs of lotteries with A-payments. This is not clear to us, as our characterization seems to be much less restrictive than existing ones for envy-freeability (which, e.g., have given rise to transforming the rent division problem to a bipartite matching computation~\cite{GMPZ17}). Furthermore, characterizations of iEF-ability with B- or C-payments are currently elusive. Our results in Section~\ref{sec:opt} reveal gaps on the quality of solutions for the two optimization problems that the different types of payments allow. Our analysis in the proofs of Theorems~\ref{thm:a-vs-b} and~\ref{thm:b-vs-a} is not tight; determining the maximum gap between A-, B, and C-payments on the quality of solutions for subsidy minimization and utility maximization is open. Finally, from the computational point of view, our solutions to subsidy minimization and utility maximization yield pairs of lotteries and C-payments that are only approximately iEF. Can these problems be solved exactly? Also, solving both subsidy minimization and utility maximization with A- or B-payments will be of practical importance. What is the complexity of these problems?




\bibliography{ief}
\end{document}